\documentclass{llncs}
\usepackage{algorithm}
\usepackage{algorithmic}
\usepackage{amssymb}
\usepackage{amsmath}
\usepackage{bigints}
\usepackage{esint}
\usepackage{array}
\usepackage{bibunits}
\usepackage{epsfig}
\usepackage{float}
\usepackage{palatino}
\usepackage{pst-all}
\usepackage{pstricks}
\usepackage{tabularx}
\usepackage{makeidx}
\usepackage{gastex}
\usepackage{listings}
\usepackage{sidecap}
\usepackage{paralist}
\usepackage{wrapfig}
\usepackage{pstricks,pst-plot,pstricks-add}
\usepackage{graphicx}
\usepackage{graphics}
\usepackage{stmaryrd}
\def\A{{\cal A}}

\def\Li{{\rm {Li}}}
\definecolor{said}{rgb}{0,0,.5}
\def\dedge{\ncline[linestyle=dotted]}
\newcommand\sym{\fam\comfam\com}
\font\tensym=msbm10 at 12pt \font\sevensym=msbm7
\font\fivesym=msbm5 
\newfam\symfam
\textfont\symfam=\tensym \scriptfont\symfam=\sevensym
\scriptscriptfont\symfam=\fivesym
\renewcommand\sym{\fam\symfam\relax}
\def\build#1_#2^#3{\mathrel{\mathop{\kern 0pt#1}\limits_{#2}^{#3}}}
\def\tend#1#2{\build\hbox to 12mm{\rightarrowfill}_{#1\rightarrow #2}^{ }}

\def\tendn{\tend{n}{\infty}}
\def\converge#1#2#3#4{\build\hbox to
#1mm{\rightarrowfill}_{#2\rightarrow #3}^{\hbox{\scriptsize #4}}}
\newcommand\Z{{\sym Z}}

\newcommand\R{{\sym R}}
\newcommand\C{{\sym C}}

\newcommand\Se{{\sym S}}
\begin{document}
%====================================================================================================================================
\frontmatter     % for the preliminaries
\pagestyle{headings} % switches on printing of running heads
\addtocmark{Hamiltonian Mechanics} % additional mark in the TOC
\lstset{language=C,basicstyle=\ttfamily\small,numbers=right,numberstyle=\tiny,frame=tb}
%====================================================================================================================================
\mainmatter
%====================================================================================================================================
\title{On the transition reduction problem for finite automata}
\titlerunning{On transition automata reduction problem}
%====================================================================================================================================
\author{el Houcein el Abdalaoui\inst{3}, Mohamed Dahmoune \inst{1} \and Djelloul Ziadi\inst{2}}
%====================================================================================================================================
\institute
{
LACL, EA 4912, Universit\'e Paris-Est Cr\'eteil (UPEC), France
%D\'epartement d'informatique,
%Facult\'e des Sciences et Technologie,
%Universit\'e Paris-Est Cr\'eteil (UPEC),
%61 avenue du G\'en\'eral de Gaulle,
%94010 Cr\'eteil Cedex,
%France.
\email{Mohamed.Dahmoune@u-pec.fr}
\and
LITIS, Universit\'e de Rouen, France
%D\'epartement d'Informatique,
%Universit\'e de Rouen,
%Avenue de l'Universit\'e,
%Technop\^ole du Madrillet,
%76801 Saint Etienne du Rouvray cedex,
%France.
\email{Djelloul.Ziadi@univ-rouen.fr}
\and
LMRS, Universit\'e de Rouen, France
\email{elhoucein.elabdalaoui@univ-rouen.fr}
}
\maketitle       % typeset the title of the contribution
%====================================================================================================================================
\begin{abstract}
%====================================================================================================================================
%$E_n=(a_1+\varepsilon)\cdot (a_2+\varepsilon)\cdot (a_3+\varepsilon)\cdots (a_n+\varepsilon)$
We are interested in the problem of transition reduction of nondeterministic automata. We present some results on the reduction of the automata recognizing the language  $L(E_n)$ denoted by the regular expression $E_n=(1+\varepsilon)\cdot (2+\varepsilon)\cdot (3+\varepsilon)\cdots (n+\varepsilon)$. These results can be used in the general case of the transition reduction problem.%\footnote{Special thanks to Sa\"id Abdedda\"im, Alexis Bes, Patrick C\'egielski and Jean-Marc Champarnaud.}
%\footnote*{Mathematical classification (2010): Primary: 68Q45, 68Q17; Secondary: 05A16.}
\end{abstract}
%==================================================================
%\subjclass{ Mathematical classification (2000): Primary: 68Q45, 68Q17; Secondary: 05A16.}
%====================================================================================================================================
\section{Introduction}\label{Intro}
%====================================================================================================================================
%\footnote{We should mentioned here that, Lifshits's proof with lower bounded \linebreak$\Omega(\frac{n\log^2 n}{\log\log~n})$ transitions \cite{Yuri} does not give rise to an algorithm.}
Minimizing the number of states of an automaton is a subject that has been studied extensively since the 1950s, both in the deterministic case and the nondeterministic case \cite{Moo,Hop}. However, works on the minimization of the number of transitions have appeared recently.

In 1997, J.~Hromkovi\~c {\it et al.} \cite{Hromkovic} have proposed an algorithm based on the concept of Common Follow Set of a regular expression, that converts a regular expression of size $n$ into a finite state automaton with $O(n)$ states, $O(n\log n)$ transitions as lower bound and $O(n \log^2 n)$ transitions as upper bound. Muscholl {\it et al.} \cite{Muscholl}, showed that this algorithm can be implemented in time $O(n \log^2 n)$. In \cite{1393793} Ouardi and Ziadi, based on the ZPC structure \cite{1516204}, gave an $O(n \log^2 n)$ algorithm to convert a weighted regular expression of size $n$ into a weighted automaton having $O(n)$ states and $O(n \log^2 n)$ transitions. In \cite{DBLP:journals/jcss/Geffert03}, Viliam Geffert showed that every regular expression of size $n$ over a fixed alphabet of $s$ symbols can be converted into a nondeterministic $\varepsilon$-free finite state automaton with $O(sn\log n)$ transitions.

Lower bound was improved by Yuri Lifshits \cite{Yuri} to $\Omega(\frac{n\log^2 n}{\log\log~n})$, after, Schnitger \cite{Schnitger06} improved it to $\Omega(n \log^2 n)$  transitions.

In \cite{Cox}, R. Cox has done an exhaustive search to find the transition minimal automata of  $L(E_n)$ for $n=1 \mbox{ to } 7$. He has also used an heuristic approach that construct transition reduced automata for $n=8 \mbox{ to } 10$.

Here, we are able to produce an algorithm for which the number of transitions is minimal for $L(E_n)$ languages class, in the sense that, asymptotically, this number of transitions is equivalent to $n \log^2 n$ (see Section~\ref{Asymptote}).

We mention that most of complexity results mentioned above are obtained from the study of  $L(E_n)$ languages class. This class of languages corresponds to a simple class of automata, in which, the minimization of the number of transitions is difficult and not obvious. The study of this class of languages, can also find its application in bioinformatics, since that $L(E_n)$ is exactly the set of all sub-sequences of the word $1.2.3\dots n$.

Our approach to reduce the number of transitions of a nondeterministic homogeneous finite state automaton is based on the decomposition of the transition table of the automaton into blocks. This decomposition is based on the concept of Common Follow Sets. From a block decomposition we construct an automaton with less transitions than the initial automaton. See Figure~\ref{F0}.%introduced by J.~Hromkovi\~c {\it et al.}\cite{Hromkovic}.

The main problem in our approach is to find a good block decomposition of the transition table (even the best one). In the case where this matrix is lower triangular or upper triangular, finding a minimal decomposition block leading to a minimal transition automaton, is not evident. Our study is focused on the upper triangular matrix, which corresponds to the transition table of the deterministic minimal automaton recognizing the  $L(E_n)$ language. The case of lower triangular matrix can be obtained in a similar manner.
%============================== Insertion d'un Tableau ========================
%\begin{figure}
%	\centering
%	\includegraphics[width=1\textwidth]{plot4.eps}
%	\caption{Experimental results using PARI/GP\cite{PARI2}.}
%	\label{fig:plot1}
%\end{figure}
%===============================================================================
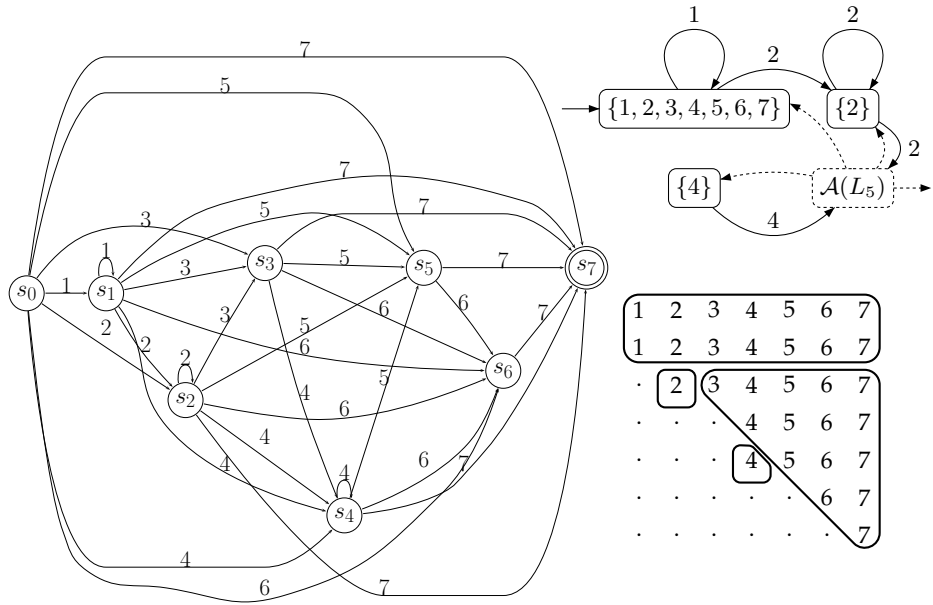
\begin{figure}[H]
\begin{minipage}{.66\textwidth}
\resizebox{\textwidth}{!}
{
\begin{pspicture}[linewidth=1bp](0bp,0bp)(636bp,593.19bp)

 \pstVerb{2 setlinejoin} % set line join style to 'mitre'
\Huge%
\psset{linecolor=black}
 % Edge: s3 -> s6
 \psbezier[arrows=->](289.72bp,351.17bp)(333.19bp,331.51bp)(445.31bp,280.78bp)(506.3bp,253.2bp)
 \psset{linecolor=[rgb]{0.0,0.0,0.0}}
 \rput(398bp,311.19bp){$6$}
 % Edge: s2 -> s3
 \psset{linecolor=black}
 \psbezier[arrows=->](196.97bp,231.3bp)(204.51bp,245.56bp)(215.7bp,266.35bp)(226bp,284.19bp)(235.71bp,301.01bp)(247.13bp,319.62bp)(261.7bp,342.95bp)
 \psset{linecolor=[rgb]{0.0,0.0,0.0}}
 \rput(230bp,304.19bp){$3$}
 % Edge: s1 -> s4
 \psset{linecolor=black}
 \psbezier[arrows=->](117.15bp,313.08bp)(119.64bp,309.97bp)(122.07bp,306.58bp)(124bp,303.19bp)(151.46bp,254.89bp)(132.2bp,228.67bp)(168bp,186.19bp)(210.27bp,136.02bp)(285.88bp,109.96bp)(336.87bp,96.388bp)
 \psset{linecolor=[rgb]{0.0,0.0,0.0}}
 \rput(230bp,147.19bp){$4$}
 % Edge: s0 -> s5
 \psset{linecolor=black}
 \psbezier[arrows=->](23.271bp,346bp)(32.821bp,397.89bp)(62.629bp,539.19bp)(104bp,539.19bp)(104bp,539.19bp)(104bp,539.19bp)(314bp,539.19bp)(343.55bp,539.19bp)(355.46bp,536.43bp)(376bp,515.19bp)(417.24bp,472.55bp)(395.4bp,443.17bp)(420bp,389.19bp)(421.47bp,385.97bp)(423.13bp,382.66bp)(429.8bp,370.6bp)
 \psset{linecolor=[rgb]{0.0,0.0,0.0}}
 \rput(230bp,547.19bp){$5$}
 % Edge: s4 -> s5
 \psset{linecolor=black}
 \psbezier[arrows=->](361.91bp,110.61bp)(376.4bp,155.8bp)(413.77bp,272.38bp)(434.1bp,335.79bp)
 \psset{linecolor=[rgb]{0.0,0.0,0.0}}
 \rput(398bp,238.19bp){$5$}
 % Edge: s1 -> s3
 \psset{linecolor=black}
 \psbezier[arrows=->](123.27bp,330.86bp)(152.56bp,336.44bp)(208.49bp,347.09bp)(252.69bp,355.51bp)
 \psset{linecolor=[rgb]{0.0,0.0,0.0}}
 \rput(188bp,354.19bp){$3$}
 % Edge: s0 -> s1
 \psset{linecolor=black}
 \psbezier[arrows=->](39.907bp,327.19bp)(50.054bp,327.19bp)(62.607bp,327.19bp)(84.236bp,327.19bp)
 \psset{linecolor=[rgb]{0.0,0.0,0.0}}
 \rput(62bp,335.19bp){$1$}
 % Edge: s2 -> s4
 \psset{linecolor=black}
 \psbezier[arrows=->](203.81bp,202.71bp)(233.32bp,181.28bp)(296.99bp,135.04bp)(340.11bp,103.72bp)
 \psset{linecolor=[rgb]{0.0,0.0,0.0}}
 \rput(272bp,175.19bp){$4$}
 % Edge: s1 -> s7
 \psset{linecolor=black}
 \psbezier[arrows=->](116.8bp,342.18bp)(123.92bp,350.21bp)(133.11bp,360.08bp)(142bp,368.19bp)(183.14bp,405.72bp)(197.79bp,415.43bp)(252bp,428.19bp)(378.33bp,457.92bp)(419.36bp,461.38bp)(544bp,425.19bp)(556.6bp,421.53bp)(560.09bp,419.79bp)(570bp,411.19bp)(579.33bp,403.09bp)(587.81bp,392.56bp)(600.26bp,374.41bp)
 \psset{linecolor=[rgb]{0.0,0.0,0.0}}
 \rput(356bp,458.19bp){$7$}
 % Edge: s5 -> s7
 \psset{linecolor=black}
 \psbezier[arrows=->](459.73bp,354.19bp)(488.6bp,354.19bp)(542.75bp,354.19bp)(588.39bp,354.19bp)
 \psset{linecolor=[rgb]{0.0,0.0,0.0}}
 \rput(524bp,362.19bp){$7$}
 % Edge: s3 -> s4
 \psset{linecolor=black}
 \psbezier[arrows=->](276.3bp,340.24bp)(282.77bp,312.37bp)(295.78bp,258.82bp)(310bp,214.19bp)(320.51bp,181.2bp)(334.88bp,144.09bp)(348.54bp,110.17bp)
 \psset{linecolor=[rgb]{0.0,0.0,0.0}}
 \rput(314bp,222.19bp){$4$}
 % Edge: s0 -> s6
 \psset{linecolor=black}
 \psbezier[arrows=->](21.9bp,308.21bp)(27.142bp,257.38bp)(42.448bp,119.19bp)(58bp,76.188bp)(65.596bp,55.187bp)(65.811bp,46.147bp)(84bp,33.188bp)(115.1bp,11.034bp)(130.27bp,21.008bp)(168bp,15.188bp)(222.58bp,6.7707bp)(239.05bp,-8.4853bp)(292bp,7.1885bp)(377.35bp,32.449bp)(396.67bp,53.646bp)(460bp,116.19bp)(488.67bp,144.51bp)(506.61bp,188.31bp)(519.03bp,226.65bp)
 \psset{linecolor=[rgb]{0.0,0.0,0.0}}
 \rput(272bp,15.188bp){$6$}
 % Edge: s4 -> s6
 \psset{linecolor=black}
 \psbezier[arrows=->](374.78bp,98.73bp)(402.46bp,109.22bp)(454.33bp,132.24bp)(486bp,167.19bp)(499.07bp,181.61bp)(508.62bp,201.34bp)(518.36bp,226.75bp)
 \psset{linecolor=[rgb]{0.0,0.0,0.0}}
 \rput(440bp,151.19bp){$6$}
 % Edge: s2 -> s5
 \psset{linecolor=black}
 \psbezier[arrows=->](205.32bp,223.81bp)(248.76bp,247.95bp)(362.55bp,311.16bp)(422.89bp,344.68bp)
 \psset{linecolor=[rgb]{0.0,0.0,0.0}}
 \rput(314bp,293.19bp){$5$}
 % Edge: s3 -> s5
 \psset{linecolor=black}
 \psbezier[arrows=->](291.6bp,358.61bp)(320.73bp,357.74bp)(375.69bp,356.1bp)(420.26bp,354.78bp)
 \psset{linecolor=[rgb]{0.0,0.0,0.0}}
 \rput(356bp,365.19bp){$5$}
 % Edge: s1 -> s6
 \psset{linecolor=black}
 \psbezier[arrows=->](122.1bp,319.94bp)(157.53bp,306.1bp)(238.77bp,276.07bp)(310bp,262.19bp)(374.95bp,249.53bp)(452.74bp,246.25bp)(504.33bp,245.27bp)
 \psset{linecolor=[rgb]{0.0,0.0,0.0}}
 \rput(314bp,270.19bp){$6$}
 % Edge: s0 -> s2
 \psset{linecolor=black}
 \psbezier[arrows=->](35.994bp,316.08bp)(48.91bp,307.14bp)(67.575bp,294.28bp)(84bp,283.19bp)(111.02bp,264.94bp)(142.22bp,244.3bp)(171.73bp,224.86bp)
 \psset{linecolor=[rgb]{0.0,0.0,0.0}}
 \rput(104bp,291.19bp){$2$}
 % Edge: s6 -> s7
 \psset{linecolor=black}
 \psbezier[arrows=->](536.24bp,260.34bp)(550.33bp,277.8bp)(573.73bp,306.78bp)(597.23bp,335.89bp)
 \psset{linecolor=[rgb]{0.0,0.0,0.0}}
 \rput(566bp,307.19bp){$7$}
 % Edge: s5 -> s6
 \psset{linecolor=black}
 \psbezier[arrows=->](452.9bp,339.85bp)(462.36bp,329.15bp)(475.35bp,314.06bp)(486bp,300.19bp)(493.47bp,290.46bp)(501.25bp,279.4bp)(513.44bp,261.38bp)
 \psset{linecolor=[rgb]{0.0,0.0,0.0}}
 \rput(482bp,317.19bp){$6$}
 % Edge: s0 -> s7
 \psset{linecolor=black}
 \psbezier[arrows=->](22.238bp,346.39bp)(29.588bp,404.93bp)(55.351bp,577.19bp)(104bp,577.19bp)(104bp,577.19bp)(104bp,577.19bp)(524bp,577.19bp)(564.33bp,577.19bp)(594.03bp,449.4bp)(607.99bp,377.21bp)
 \psset{linecolor=[rgb]{0.0,0.0,0.0}}
 \rput(314bp,585.19bp){$7$}
 % Edge: s4 -> s7
 \psset{linecolor=black}
 \psbezier[arrows=->](375.78bp,93.824bp)(397.47bp,96.267bp)(432.82bp,102.19bp)(460bp,116.19bp)(473.72bp,123.25bp)(475.56bp,127.83bp)(486bp,139.19bp)(515.22bp,171bp)(521.42bp,180.36bp)(544bp,217.19bp)(565.48bp,252.22bp)(585.84bp,294.96bp)(602.76bp,332.67bp)
 \psset{linecolor=[rgb]{0.0,0.0,0.0}}
 \rput(482bp,147.19bp){$7$}
 % Edge: s2 -> s6
 \psset{linecolor=black}
 \psbezier[arrows=->](207.32bp,210.02bp)(231.27bp,205.09bp)(273.42bp,197.17bp)(310bp,194.19bp)(378.61bp,188.6bp)(456.68bp,216.16bp)(506.26bp,236.97bp)
 \psset{linecolor=[rgb]{0.0,0.0,0.0}}
 \rput(356bp,205.19bp){$6$}
 % Edge: s1 -> s1
 \psset{linecolor=black}
 \psbezier[arrows=->](96.502bp,345.08bp)(95.102bp,355.03bp)(97.602bp,364.19bp)(104bp,364.19bp)(108.1bp,364.19bp)(110.6bp,360.43bp)(111.5bp,345.08bp)
 \psset{linecolor=[rgb]{0.0,0.0,0.0}}
 \rput(104bp,372.19bp){$1$}
 % Edge: s0 -> s3
 \psset{linecolor=black}
 \psbezier[arrows=->](30.783bp,343.17bp)(42.034bp,358.28bp)(61.201bp,379.96bp)(84bp,389.19bp)(138.42bp,411.22bp)(207.83bp,388.24bp)(254.45bp,367.97bp)
 \psset{linecolor=[rgb]{0.0,0.0,0.0}}
 \rput(146bp,405.19bp){$3$}
 % Edge: s3 -> s7
 \psset{linecolor=black}
 \psbezier[arrows=->](285.01bp,373.67bp)(300.34bp,389.06bp)(327.36bp,411.19bp)(356bp,411.19bp)(356bp,411.19bp)(356bp,411.19bp)(524bp,411.19bp)(549.19bp,411.19bp)(572.85bp,394.54bp)(596.75bp,371.91bp)
 \psset{linecolor=[rgb]{0.0,0.0,0.0}}
 \rput(440bp,419.19bp){$7$}
 % Edge: s2 -> s2
 \psset{linecolor=black}
 \psbezier[arrows=->](180.5bp,232.08bp)(179.1bp,242.03bp)(181.6bp,251.19bp)(188bp,251.19bp)(192.1bp,251.19bp)(194.6bp,247.43bp)(195.5bp,232.08bp)
 \psset{linecolor=[rgb]{0.0,0.0,0.0}}
 \rput(188bp,259.19bp){$2$}
 % Edge: s1 -> s5
 \psset{linecolor=black}
 \psbezier[arrows=->](119.98bp,338.42bp)(132.58bp,346.9bp)(150.83bp,358.37bp)(168bp,366.19bp)(228.26bp,393.61bp)(244.57bp,401.07bp)(310bp,411.19bp)(349.87bp,417.35bp)(391.9bp,392.33bp)(425.19bp,367.13bp)
 \psset{linecolor=[rgb]{0.0,0.0,0.0}}
 \rput(272bp,416.19bp){$5$}
 % Edge: s1 -> s2
 \psset{linecolor=black}
 \psbezier[arrows=->](113.21bp,310.33bp)(120.45bp,297.55bp)(131.12bp,279.75bp)(142bp,265.19bp)(149.73bp,254.85bp)(159.15bp,244.13bp)(174.38bp,227.88bp)
 \psset{linecolor=[rgb]{0.0,0.0,0.0}}
 \rput(146bp,273.19bp){$2$}
 % Edge: s0 -> s4
 \psset{linecolor=black}
 \psbezier[arrows=->](21.365bp,308.16bp)(26.509bp,243.64bp)(47.39bp,37.188bp)(104bp,37.188bp)(104bp,37.188bp)(104bp,37.188bp)(272bp,37.188bp)(297.17bp,37.188bp)(320.66bp,54.682bp)(343.43bp,77.16bp)
 \psset{linecolor=[rgb]{0.0,0.0,0.0}}
 \rput(188bp,45.188bp){$4$}
 % Edge: s4 -> s4
 \psset{linecolor=black}
 \psbezier[arrows=->](348.5bp,110.08bp)(347.1bp,120.03bp)(349.6bp,129.19bp)(356bp,129.19bp)(360.1bp,129.19bp)(362.6bp,125.43bp)(363.5bp,110.08bp)
 \psset{linecolor=[rgb]{0.0,0.0,0.0}}
 \rput(356bp,137.19bp){$4$}
 % Edge: s2 -> s7
 \psset{linecolor=black}
 \psbezier[arrows=->](199.42bp,198.29bp)(234.82bp,149.82bp)(343.61bp,7.1885bp)(398bp,7.1885bp)(398bp,7.1885bp)(398bp,7.1885bp)(524bp,7.1885bp)(588.36bp,7.1885bp)(606.18bp,234.31bp)(611.13bp,331.15bp)
 \psset{linecolor=[rgb]{0.0,0.0,0.0}}
 \rput(398bp,16.188bp){$7$}
 % Node: s3
{%
 \psset{linecolor=[rgb]{0.0,0.0,0.0}}
 \psellipse[](272bp,359bp)(19bp,19bp)
 \rput(272bp,359.19bp){$s_{3}$}
}%
 % Node: s2
{%
 \psset{linecolor=[rgb]{0.0,0.0,0.0}}
 \psellipse[](188bp,214bp)(19bp,19bp)
 \rput(188bp,214.19bp){$s_{2}$}
}%
 % Node: s1
{%
 \psset{linecolor=[rgb]{0.0,0.0,0.0}}
 \psellipse[](104bp,327bp)(19bp,19bp)
 \rput(104bp,327.19bp){$s_{1}$}
}%
 % Node: s0
{%
 \psset{linecolor=[rgb]{0.0,0.0,0.0}}
 \psellipse[](20bp,327bp)(19bp,19bp)
 \rput(20bp,327.19bp){$s_{0}$}
}%
 % Node: s7
{%
 \psset{linecolor=[rgb]{0.0,0.0,0.0}}
 \psellipse[](612bp,354bp)(19bp,19bp)
 \psellipse[](612bp,354bp)(23bp,23bp)
 \rput(612bp,354.19bp){$s_{7}$}
}%
 % Node: s6
{%
 \psset{linecolor=[rgb]{0.0,0.0,0.0}}
 \psellipse[](524bp,245bp)(19bp,19bp)
 \rput(524bp,245.19bp){$s_{6}$}
}%
 % Node: s5
{%
 \psset{linecolor=[rgb]{0.0,0.0,0.0}}
 \psellipse[](440bp,354bp)(19bp,19bp)
 \rput(440bp,354.19bp){$s_{5}$}
}%
 % Node: s4
{%
 \psset{linecolor=[rgb]{0.0,0.0,0.0}}
 \psellipse[](356bp,92bp)(19bp,19bp)
 \rput(356bp,92.188bp){$s_{4}$}
}%
\end{pspicture}
}
\end{minipage}
\hfill
\begin{minipage}{.31\textwidth}
{
\unitlength=3pt
\begin{picture}(25,25)(-7,-15)
\gasset{Nadjust=wh,Nmr=1}
\node[Nmarks=i](q0)(0,10){$\{1,2,3,4,5,6,7\}$}
\node(q1)(20,10){$\{2\}$}
\node(q2)(0,0){$\{4\}$}
\node[dash={0.4 0.4}0,Nmarks=f](q3)(20,0){$\A(L_5)$}
\drawloop(q0){$1$}
\drawloop(q1){$2$}
\drawedge[curvedepth=5](q0,q1){$2$}
\drawedge[curvedepth=6](q1,q3){$2$}
\drawedge[curvedepth=-6](q2,q3){$4$}
\drawedge[dash={0.4 0.4}0,curvedepth=-6](q3,q0){}
\drawedge[dash={0.4 0.4}0,curvedepth=-4](q3,q1){}
\drawedge[dash={0.4 0.4}0,curvedepth=-2](q3,q2){}
\end{picture}
}
{
\begin{pspicture}(3,3)
\rput(0.0,3.0){1}
\rput(0.5,3.0){2}
\rput(1.0,3.0){3}
\rput(1.5,3.0){4}
\rput(2.0,3.0){5}
\rput(2.5,3.0){6}
\rput(3.0,3.0){7}

\rput(0.0,2.5){1}
\rput(0.5,2.5){2}
\rput(1.0,2.5){3}
\rput(1.5,2.5){4}
\rput(2.0,2.5){5}
\rput(2.5,2.5){6}
\rput(3.0,2.5){7}

\rput(0.0,2.0){.}
\rput(0.5,2.0){2}
\rput(1.0,2.0){3}
\rput(1.5,2.0){4}
\rput(2.0,2.0){5}
\rput(2.5,2.0){6}
\rput(3.0,2.0){7}

\rput(0.0,1.5){.}
\rput(0.5,1.5){.}
\rput(1.0,1.5){.}
\rput(1.5,1.5){4}
\rput(2.0,1.5){5}
\rput(2.5,1.5){6}
\rput(3.0,1.5){7}

\rput(0.0,1.0){.}
\rput(0.5,1.0){.}
\rput(1.0,1.0){.}
\rput(1.5,1.0){4}
\rput(2.0,1.0){5}
\rput(2.5,1.0){6}
\rput(3.0,1.0){7}

\rput(0.0,0.5){.}
\rput(0.5,0.5){.}
\rput(1.0,0.5){.}
\rput(1.5,0.5){.}
\rput(2.0,0.5){.}
\rput(2.5,0.5){6}
\rput(3.0,0.5){7}

\rput(0.0,0.0){.}
\rput(0.5,0.0){.}
\rput(1.0,0.0){.}
\rput(1.5,0.0){.}
\rput(2.0,0.0){.}
\rput(2.5,0.0){.}
\rput(3.0,0.0){7}

\pspolygon[linearc=0.1](0.25,2.2)(0.75,2.2)(0.75,1.7)(0.25,1.7)%2
%\pspolygon[linearc=0.1](1.25,1.2)(1.75,1.2)(1.75,0.7)(1.25,0.7)%4
\pspolygon[linearc=0.1](1.25,1.2)(1.5,1.19)(1.75,0.93)(1.75,0.7)(1.25,0.7)%4

\pspolygon[linearc=0.2](-0.2,3.2)(3.2,3.2)(3.2,2.3)(-0.2,2.3)
\pspolygon[linearc=0.2](0.55,2.2)(3.2,-0.45)(3.2,2.2)
\end{pspicture}
}
\end{minipage}
\caption{The reduced automaton (at right) is obtained from a decomposition transition table (at right bottom) of the homogeneous automaton (at left). The automaton $\A(L_5)$ is the part of the reduced automaton which represents the triangle (in the transition table decomposition).}
\label{F0}
\end{figure}
In this paper we present the following results: At first, in Section~\ref{CFS}, we extend the concept of Common Follow Sets to homogeneous automata. Then, in Section~\ref{Reduction} we introduce particular decompositions called $Z$-partitions associated with expressions $E_n$. Then, in Section~\ref{TreeReduction} we introduce the notion of $Z$-tree to represent any $Z$-partition by a binary tree. Then, we propose an algorithm of $ O (n \log n) $ time complexity to generate the $Z$-minimal trees. We finish our study by experimental results and a last section in which we show that our algorithms construct automata with a number of transitions equivalents to $n\log^2(n)$ which is the minimal lower bound according to Schnitger.
%Our approach of the transition reduction is based on this algorithm.
%===================================================================================================================================
\section{Notation and terminology}\label{Notation}
%===================================================================================================================================
%\subsection{Regular expressions}
%===================================================================================================================================
We recall the basics of regular expressions, languages and finite state machines and introduce the notation that we use. Let $\Sigma$ be a non-empty finite set of symbols, called alphabet. The set of all the words over $\Sigma$ is denoted by $\Sigma^*$. The empty word is denoted by $\varepsilon$. A language over $\Sigma$ is a subset of $\Sigma^*$. A finite automaton over $\Sigma$ is a 5-tuple $\A = (Q,\Sigma,I,\delta,F)$ where $Q$ is a set of states, $I$ is a subset of $Q$ whose elements are the initial states, $F$ is a subset of $Q$ whose elements are the final states, $\delta$ is a subset of the cartesian product $Q\times \Sigma \times Q$ whose elements are the transitions. A transition $(q,a,p)\in \delta$ goes from the head $q$ to the tail $p$. A path in $\A$ is a sequence of transitions $(q_i,a_i,q_{i+1})$, $i=1 \mbox{ to } n$, of consecutive transitions. Its label is the word $w=a_1a_2\cdots a_n$. A word $w \in \Sigma^*$ is recognized by the automaton $\A$ if there is a path with label $w$ such that $q_1\in I$ and $q_{n+1}\in F$. The language recognized by the automaton $\A$ is the set of words that are recognized by $\A$. The automaton $\A$ is homogeneous if for all $(q,a,p),\ (q',a',p')\in \delta$, $p=p'$ implies that $a=a'$, in this case we write $h(p)=a$. The function $h$ assigns to each non-initial state $q$ of an homogeneous automaton the symbol that is the unique label of all the transitions having $q$ as tail.%By convention the symbol associated with an initial state is $\bot$.

In Appendix~\ref{AsymptoticNotations} we recall the basics of asymptotic notations.
%===================================================================================================================================
\section{CFS for homogeneous automata}\label{CFS}
%===================================================================================================================================
J.~Hromkovi\~c {\it et al.} \cite{Hromkovic} have given an elegant algorithm based on the notion of Common Follow Sets, to convert a regular expression of size $n$ into a nondeterministic finite automaton having $O(n)$ states and $O(n\log^2n)$ transitions. This notion can be easily extended to homogeneous automata.

Let $\A=(Q,\Sigma,\{q_0\},\delta,F)$ be an homogeneous automaton.
%The automaton $\A$ can be seen as vertex-labeled graph $G_{\A}=(X,U)$ where $X=Q$ and $U=\{(p,q)\;\mid\; (p,a,q)\in \delta \}$.
In order to capture the final states in the $\A$, we introduce a dummy state denoted by $\#$ which is not in $Q$. We define over $Q$ the function $follow$ as follows:
\begin {eqnarray*}
 follow(q)&=&
 \left\{
 \begin{array}{l}
  \{p\;\mid\; (q,a,p) \in \delta\} \cup \{\#\} \mbox{ if } q \in F,\\
  \{p\;\mid\; (q,a,p) \in \delta\} \mbox{ otherwise. }
\end{array} \right.
\end {eqnarray*}
%$follow(q)$ is the set of successors of $q$ in the graph $G_{\A}$.
Let $q \in Q$ be a state in $\A$, we denote by $dec(q)=\{Q_1,Q_2, \cdots, Q_k\}$ (where $Q_i \subseteq follow(q)$) any decomposition of the set $follow(q)$, i.e. $\displaystyle follow(q)=\bigcup_{Q_i \in dec(q)} Q_i$. In the case where $dec(q)$ is a partition of the set $follow(q)$,
%, i.e. $follow(q)=\biguplus\limits_{Q_i \in dec(q)}Q_i$,
the decomposition $dec(q)$ will be called a partition decomposition. Figure~\ref{F1} provides examples of decompositions.
\begin{figure}
\begin{center}
\unitlength=3pt
\resizebox{.55\textwidth}{!}
{
\begin{picture}(75,10)(0,0)
\gasset{Nadjust=w,Nadjustdist=3}
\node[Nmarks=r](q0)(10,0){$0$}
\imark[ilength=5](q0)
\node[Nmarks=r](q1)(30,0){$1$}
\node[Nmarks=r](q2)(50,0){$2$}
\node[Nmarks=r](q3)(70,0){$3$}
\drawedge(q0,q1){$1$}
\drawedge[curvedepth=7](q0,q2){$2$}
\drawedge[curvedepth=14](q0,q3){$3$}
\drawedge(q1,q2){$2$}
\drawedge[curvedepth=-9](q1,q3){$3$}
\drawedge(q2,q3){$3$}
\end{picture}
}
\end{center}
\caption[soubhanALLAH]
{We have $follow (0)=\{1,2,3,\#\}$. Here are three possible decompositions of $follow(0)$. The two first ones are partition decompositions.
\begin{inparaenum}[(i)]
\item $dec(0)= \{\{1,2\},\{3,\#\}\}$
\item $dec(0)= \{\{1\},\{2\},\{3,\#\}\}$
\item $dec(0)= \{\{1,2\},\{2,3,\#\}\}$
\end{inparaenum}.}
\label{F1}
\end{figure}
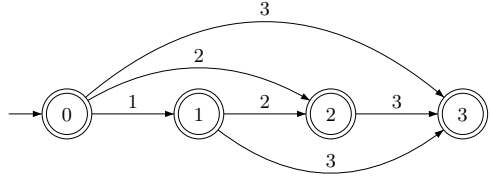
\begin{definition}[Common Follow Sets System]
Let $\A$ be a homogeneous automaton. A CFS system for $\A$ is given as $S(\A)=(dec(q))_{q\in Q}$, where each
 $dec(q)\subseteq 2^Q$ is a decomposition of $follow(q)$.
\end{definition}
\begin{definition}[CFS automaton]
Let $\A=(Q,\Sigma,\{q_0\},\delta,F)$ be a homogeneous automaton and $S(\A)$ an associated Common Follow Sets system. The Common Follow Sets automaton associated with $S(\A)$ is defined by ${\cal C}_{S(\A)}=(Q',\Sigma,I',\delta',F')$ where
\begin{itemize}
%\item $Q'=\bigcup\limits_{q\in Q\cup \{\#\} } dec(q)$ % {\cal F}_{\A}$
\item $Q'=\bigcup\limits_{q\in Q} dec(q)$ % {\cal F}_{\A}$
\item $I'= dec(q_0)$
\item $\mbox{ For } Q_1\in Q',\;Q_1\in F'$ if and only if $\#\in Q_1$
\item $\delta' = \{(Q_1,a,Q_2)\; \mid\; \exists q\in Q_1 \mbox{ s.t. } h(q)=a \mbox { and } Q_2\in dec(q) \}$.
\end{itemize}
\end{definition}
\begin{theorem}
Let $\A$ be a homogeneous automaton, $S(\A)$ be a Common Follow Sets System associated with $\A$ and ${\cal C}_{S(\A)}$ its Common Follow Sets automaton. Then ${\cal C}_{S(\A)}$ and $\A$ recognize the same language.
\end{theorem}
This theorem can be proved in the same way as Theorem~5 of the paper of J.~Hromkovi\~c {\it et al.} \cite{Hromkovic}.

To evaluate the number of transitions in the automaton ${\cal C}_{S(\A)}$ we define over the states of $\A$ two functions,
\begin{itemize}
\item $a(q)=|dec(q)|$ the size of the decomposition of the set $follow(q)$
\item $b(q)=|\{ Q_1 \in Q' \; \mid\; q\in Q_1\}|$ the number of states in $Q'$ that contain the state $q$.
\end{itemize}
\begin{lemma}
\label{n_trans}
The number of transitions $T_{\cal C}$ in ${\cal C}_{S(\A)}$ is such that $\displaystyle T_{\cal C} \leq \sum_{q\in Q} a(q)b(q)$.
\end{lemma}
It is easy to see that if for all $p,q\in Q \backslash \{q_0\}$ such that $p\neq q$, we have $h(p) \neq h(q)$ then the equality holds. From Lemma~\ref{n_trans}, we can see that the number of transitions in a CFS automaton depends on the decomposition system. A decomposition which is not a partition will induce more transitions than a partition decomposition. Therefore in the following we are interested only in partition decompositions. As it was mentioned in the introduction our study will focus on the CFS automata associated with the family of automata $(\A_n)_{n\ge 1}$. The automaton $\A_n=(Q, \Sigma ,I,\delta ,F)$ is defined by:
\begin{itemize}
\item $\Sigma=\{1,2,\dots,n\}$
\item $Q=\Sigma \cup \{0\}$
\item $F=Q$
\item $I=\{0\}$, $\delta=\{ (p,q,q)\in Q \times \Sigma \times Q \; \mid \; q>p\}$.
\end{itemize} Figure~\ref{F2} shows two CFS automata associated with the automaton $\A_3$.

%%%In~\cite{Yuri} Lifshits has shown that any automaton recognizing $ E_n $ has at least $\Omega(\frac{n\log^2 n}{\log\log~n})$ transitions. Unfortunately the proof of this result does not give rise to an algorithm.
In the next sections we present two algorithms that construct particular CFS systems which correspond to CFS automata with a reduced number of transitions. In the last section  we give comparative and experimental results.
%%%In the last section we give a comparison of this number with the lower bound established by Lifshits.
\begin{figure}
\begin{minipage}{.46\textwidth}
\begin{center}
\resizebox{.9\textwidth}{!}{
\begin{picture}(80,20)(0,0)
\gasset{Nadjust=w,Nadjustdist=3}
\node[Nmarks=r](q0)(10,0){$\{1,2,3,\#\}$}
\imark[ilength=5](q0)
\node[Nmarks=r](q2)(25,20){$\{2,3,\#\}$}
\node[Nmarks=r](q3)(50,0){$\{3,\#\}$}
\node[Nmarks=r](q4)(75,0){$\{\#\}$}
\drawedge(q0,q3){$2$}
\drawedge(q0,q2){$1$}
\drawedge(q2,q3){$2$}
\drawedge(q3,q4){$3$}
\drawedge(q0,q3){$2$}
\drawedge[curvedepth=-8](q0,q4){$3$}
\drawedge[curvedepth=4](q2,q4){$3$}
\end{picture}
}

\[
\begin{array}{|c|c|c|c|c|}\hline
 q & a(q) & b(q) & a(q)b(q) & dec(q)\\ \hline
 0 & 1 & 0 & 0 & \{ \{1,2,3,\#\} \}\\ \hline
 1 & 1 & 1 & 1 &\{ \{2,3,\#\} \}\\ \hline
 2 & 1 & 2 & 2&\{ \{3,\#\} \}\\ \hline
 3 & 1 & 3 & 3&\{ \{\#\} \}\\ \hline
\end{array}
\]

$T_{\cal C} = 6$
\end{center}
%\caption{CFS automaton associated with the following CFS system. It has one initial state and $T_{\cal C} = 6$ transitions}
%\label{fig1}
\end{minipage}
\hfill
\begin{minipage}{.46\textwidth}
\begin{center}
\resizebox{.9\textwidth}{!}{
\begin{picture}(80,20)(0,0)
\gasset{Nadjust=w,Nadjustdist=3}
\node[Nmarks=r](q0)(10,0){$\{3,\#\}$}
\imark[ilength=5](q0)
\node(q1)(10,20){$\{1,2\}$}
\imark[ilength=5](q1)
\node(q2)(25,10){$\{2\}$}
\node(q3)(50,10){$\{3\}$}
\node[Nmarks=r](q4)(75,10){$\{\#\}$}

\drawedge(q1,q0){$1$}
\drawedge[curvedepth=-4](q1,q2){$1$}
\drawedge[curvedepth=4](q1,q3){$2$}
\drawedge[curvedepth=8](q1,q4){$2$}

\drawedge[curvedepth=-8](q2,q4){$2$}
\drawedge(q2,q3){$2$}
\drawedge(q3,q4){$3$}
\drawedge[curvedepth=-12](q0,q4){$3$}
\end{picture}
}

\[
\begin{array}{|c|c|c|c|c|}\hline
 q & a(q) & b(q) & a(q)b(q) & dec(q)\\ \hline
 0 & 2 & 0 & 0&\{ \{1,2\} ,\{3,\#\} \}\\ \hline
 1 & 2 & 1 & 2&\{ \{2\},\{3,\#\} \}\\ \hline
 2 & 2 & 2 & 4&\{ \{3\},\{\#\} \}\\ \hline
 3 & 1 & 2 & 2&\{ \{\#\} \}\\ \hline
\end{array}
\]

$T_{\cal C} = 8$
\end{center}
%\caption{CFS automaton associated with the following CFS system. It has two initial states and $T_{\cal C} = 8$ transitions}
%\label{fig2}
\end{minipage}
%\begin{minipage}{.46\textwidth}
%\begin{eqnarray*}
%dec(0)&=&\{ \{1,2,3,\#\} \}\\
%dec(1)&=&\{ \{2,3,\#\} \}\\
%dec(2)&=&\{ \{3,\#\} \}\\
%dec(3)&=&\{ \{\#\} \}\\
%\end{eqnarray*}
%
%\end{minipage}
%\hfill
%\begin{minipage}{.46\textwidth}
%\begin{eqnarray*}
%dec(0)&=&\{ \{1,2\} ,\{3,\#\} \}\\
%dec(1)&=&\{ \{2\},\{3,\#\} \}\\
%dec(2)&=&\{ \{3\},\{\#\} \}\\
%dec(3)&=&\{ \{\#\} \}\\
%\end{eqnarray*}
%\end{minipage}

%\begin{minipage}{.46\textwidth}
%\[
%\begin{array}{|c|c|c|c|c|}\hline
% q & a(q) & b(q) & a(q)b(q) & dec(q)\\ \hline
% 0 & 1 & 0 & 0 & \{ \{1,2,3,\#\} \}\\ \hline
% 1 & 1 & 1 & 1 &\{ \{2,3,\#\} \}\\ \hline
% 2 & 1 & 2 & 2&\{ \{3,\#\} \}\\ \hline
% 3 & 1 & 3 & 3&\{ \{\#\} \}\\ \hline
%\end{array}
%\]
%\end{minipage}
%\hfill
%\begin{minipage}{.46\textwidth}
%\[
%\begin{array}{|c|c|c|c|c|}\hline
% q & a(q) & b(q) & a(q)b(q) & dec(q)\\ \hline
% 0 & 2 & 0 & 0&\{ \{1,2\} ,\{3,\#\} \}\\ \hline
% 1 & 2 & 1 & 2&\{ \{2\},\{3,\#\} \}\\ \hline
% 2 & 2 & 2 & 4&\{ \{3\},\{\#\} \}\\ \hline
% 3 & 1 & 2 & 2&\{ \{\#\} \}\\ \hline
%\end{array}
%\]
%\end{minipage}
\caption{Two CFS Automata constructed from the automaton $\A_3$ shown in Figure~\ref{F1}.}
\label{F2}
\end{figure}
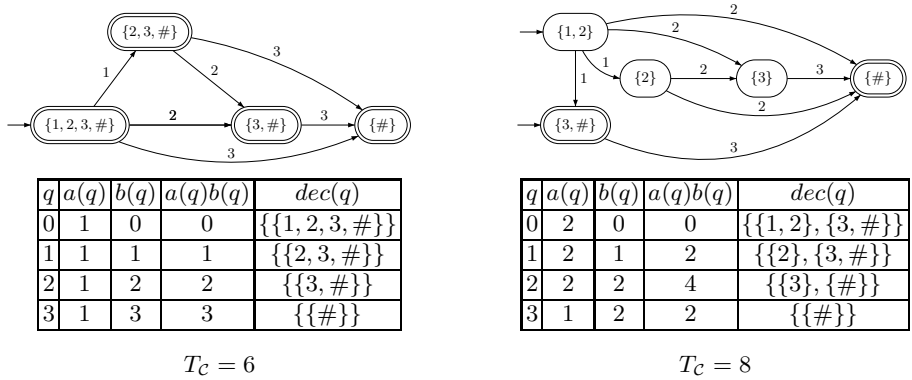
%===================================================================================================================================
\section{$\A_n$ Reduction Algorithm}\label{Reduction}
%===================================================================================================================================
The notion of \emph{$Z$-partition} is nowhere introduced formally. In the following we are interested in reducing the number of transitions in the automaton $\A_n$. The following algorithm
computes particular CFS systems $S(\A_n)$ that provide ${\cal C}_{S(\A_n)}$ automata with small number of transitions and having $n+1$ states. Let $E_n=(1+\varepsilon)\cdot (2+\varepsilon)\cdot (3+\varepsilon)\cdots (n+\varepsilon)$ be a regular expression, it is easy to see that the language denoted by the expression $E_n$ is exactly the language recognized by the automaton $\A_n$. We have:
\begin{proposition}
Each transition minimal automaton that recognizes $L(E_n)$ has exactly $n+1$ states.
\end{proposition}
This proposition can be proved using properties of the universal automaton~\cite{UniversalAutomaton} of $L(E_n)$.
%\footnote{See \cite{UniversalAutomaton} for more information on the universal automaton.}
%~\ref{alg:alg1}

For fixed $n$, a CFS system produced by the following algorithm will be denoted by $Z(\A_n)$. The set of all $Z(A_n)$ will be denoted $CFSZ(n)$. Our aim in this section is to compute all minimal decompositions in $CFSZ(n)$.
\begin{algorithm}[H]
\caption{CFSPartitions$(n)$}
\label{alg:alg1}
\begin{algorithmic}[1]
\REQUIRE $n \in \mathbb{N}$
\ENSURE $Z(\A_n)$
\STATE $Q\leftarrow\{0,1,2,3,4, \dots ,n\}$% \COMMENT{}
\FOR{$i$ = $0$ to $n$}
 \STATE $follow(i) \leftarrow\{j \in Q|j>i\} \cup \{\#\}$
 \STATE $dec(i) \leftarrow \phi$
 \STATE $Q_i \leftarrow \phi$
\ENDFOR
\FOR{$i$ = $0$ to $n$}
 \STATE Choose $j$ in $Q$
 \STATE $Q \leftarrow Q \backslash \{j\}$
 \STATE $Q_j \leftarrow follow(j)$
 \FORALL{$k \in Q$}
 %\IF{($follow(k)\cap Q_j\neq \phi$)  $follow(k) \supseteq Q_j$}
 \IF{($Q_j\subseteq follow(k)$)}
  \STATE $dec(k) \leftarrow dec(k)\cup\{Q_j\}$
  \STATE $follow(k) \leftarrow follow(k) \backslash Q_j$
 %\ELSE
 \ENDIF
 \ENDFOR
\ENDFOR
%\RETURN
\end{algorithmic}
\end{algorithm}
\begin{proposition}
The number of all $Z(\A_n)$ CFS partition systems is the $n^{th}$ Catalan number: $|CFSZ(n)|=\frac{1}{n+1}{2n\choose n}$.
\end{proposition}
The successive choice of values of $j$ (line 8) leads to a permutation of size $n$. So, each CFS partition system $Z(\A_n)$ can be associated with at least one permutation of size $n$.

The following algorithm is a recursive version of Algorithm~\ref{alg:alg1}. Its first call is done by RecursiveDecomposition$(0,n)$. Without loss of generality we associate in this last algorithm the dummy state $\#$ to the number $n+1$. At each call, Algorithm~\ref{alg:alg2} constructs one block from the transition matrix $M$, for the call RecursiveDecomposition$(n_1,n_2)$ and the choice of $j$ (line 2), it produces the block $B_j$ which is the submatrix $M[j..n_1;j..n_2]$.% with lower left  corner $M[j,j]$ and upper right corner $M[n_1,n_2]$.
\begin{algorithm}[H]
\caption{RecursiveDecomposition $(n_1,n_2)$}
\label{alg:alg2}
\begin{algorithmic}[1]
\REQUIRE $ n_1, n_2 \in \mathbb{N}$
\ENSURE $Z(\A_n)$ when $n_1=0$ and $n_2=n$
\IF{$ n_1 \leq n_2 $}
 \STATE Choose an integer $ j $ between $ n_1 $ and $ n_2 $, $ j \in \{n_1, \dots ,n_2\} $
 \STATE $Q_j\leftarrow\{j+1, \dots ,n_2+1\}$
 \FOR{$k$ = $n_1$ to $j$}
 \STATE $dec(k) \leftarrow dec(k)\cup\{Q_j\}$
 \ENDFOR
 \STATE $B_j=M[j..n_1;j..n_2]$
 \STATE RecursiveDecomposition$(n_1,j-1)$
 \STATE RecursiveDecomposition$(j+1,n_2)$
%\ELSE
\ENDIF
\end{algorithmic}
\end{algorithm}
\begin{example}
In this example we shows the CFS partition systems associated with permutations $(0,2,1,3)$ and permutation $(1,0,2,3)$.

\begin{center}
\begin{minipage}{.45\textwidth}
\resizebox{1.00\textwidth}{!}{%
\begin{tabular}{ccccccccccccc}
$dec(0)=\{$&$\{$&1&,&2& &,& &3&,&$\#$&$\}$&$\}$\\
$dec(1)=\{$& & &$\{$&2&$\}$&,&$\{$&3&,&$\#$&$\}$&$\}$\\
$dec(2)=\{$& & & & & & &$\{$&3&,&$\#$&$\}$&$\}$\\
$dec(3)=\{$& & & & & & & & &$\{$&$\#$&$\}$&$\}$
\end{tabular}
\begin{tabular}{ccc|cc|}
\cline{2-5}
\multicolumn{1}{c|}{$0:$} & $1$ & \multicolumn{1}{c}{$2$} & $3$ & $\#$\tabularnewline
\cline{2-5}
$1:$ & \multicolumn{1}{c|}{} & $2$ & $3$ & $\#$\tabularnewline
\cline{3-3}
$2:$ & & & $3$ & $\#$\tabularnewline
\cline{4-5}
$3:$ & & \multicolumn{1}{c}{} & \multicolumn{1}{c|}{} & $\#$\tabularnewline
\cline{5-5}
\end{tabular}}
\end{minipage}
\hfill
\begin{minipage}{.45\textwidth}
\resizebox{1.00\textwidth}{!}{%
\begin{tabular}{ccccccccccccc}
$dec(0)=\{$&$\{$&1&$\}$&,&$\{$&2&,&3&,&\#&$\}$&$\}$\\
$dec(1)=\{$& & & & &$\{$&2&,&3&,&\#&$\}$&$\}$\\
$dec(2)=\{$& & & & & & &$\{$&3&,&\#&$\}$&$\}$\\
$dec(3)=\{$& & & & & & & & &$\{$&\#&$\}$&$\}$
\end{tabular}
\begin{tabular}{cc|ccc|}
\cline{2-5}
\multicolumn{1}{c|}{$0:$} & $1$ & $2$ & $3$ & $\#$\tabularnewline
\cline{2-2}
$1:$ & & $2$ & $3$ & $\#$\tabularnewline
\cline{3-5}
$2:$ & \multicolumn{1}{c}{} & \multicolumn{1}{c|}{} & $3$ & $\#$\tabularnewline
\cline{4-5}
$3:$ & \multicolumn{1}{c}{} & & \multicolumn{1}{c|}{} & $\#$\tabularnewline
\cline{5-5}
\end{tabular}}
\end{minipage}
\end{center}

For permutation $(0,2,1,3)$ the first block $B_0$ is
\begin{tabular}{|c|c|c|c|}
\hline
$1$ & $2$ & $3$ & $\#$\\
\hline
\end{tabular}
, the second block $B_2$ is
\begin{tabular}{|cc|}
\hline
$3$ & $\#$\\
$3$ & $\#$\\
\hline
\end{tabular},
the third block $B_1$ is
\begin{tabular}{|c|}
\hline
$2$\\
\hline
\end{tabular} and $B_3=$
\begin{tabular}{|c|}
\hline
$\#$\\
\hline
\end{tabular} is the last one.
\end{example}
\begin{proposition}
The computation of all minimal partition system $Z(\A_n)$ in $CFSZ(n)$ can be done in time $O(n!)$.
\end{proposition}
\begin{proof}
This can be done by calling the nondeterministic Algorithm~\ref{alg:alg1} or \ref{alg:alg2} for each possible execution.
\end{proof}
\begin{remark}
By the use of the dynamic programming, we can improve the exponential brute force method to a polynomial algorithm as shown in Algorithm~\ref{alg:alg3}.
\label{remark2}
\end{remark}
%\begin{algorithm}[H]
%\caption{DynamicDecomposition $(n)$}
%\label{alg:alg3}
%\begin{algorithmic}[1]
%\REQUIRE $ n \in \mathbb{N}$
%\ENSURE $Z(\A_n)$ which can easily deduced from the array $T$ and the three-dimensional table $Z$, see Section~\ref{A}.
%\STATE\COMMENT{Initialization}
%\FOR{$p$ = $0$ to $n$}
%	\FOR{$q$ = $0$ to $n$}
% 		\STATE $Z[1][p][q] \leftarrow p*q$
%	\ENDFOR
%\ENDFOR
%\STATE\COMMENT{Main}
%\FOR{$i$ = $1$ to $n$}
%\FOR{$p$ = $0$ to $n$}
%\FOR{$q$ = $0$ to $n$}
%\FOR{$k$ = $0$ to $i$}
%\STATE $T[k]\leftarrow Z[k][p+1][q]+Z[i-k-1][p][q+1]$
%\ENDFOR
%\STATE $\displaystyle Z[i][p][q] \leftarrow \min_{k \in \{0,\dots,i\}} T[k]$
%\ENDFOR
%\ENDFOR
%\ENDFOR
%\end{algorithmic}
%\end{algorithm}
%\begin{theorem}
%Algorithm~\ref{alg:alg3} computes all minimal partition systems $Z(\A_i)$ in $CFSZ(i)$ for all $i=1 \mbox{ to } n$ in $O(n^4)$ time and $O(n^3)$ space.
%\end{theorem}
%\begin{proof}
%%The proof of this theorem is detailed in the next section.
%The proof of this theorem is detailed in Appendix~\ref{A}.
%\end{proof}
%%See Appendix~\ref{A}.
%%Notice that Algorithm~\ref{alg:alg1} and  Algorithm~\ref{alg:alg2} build only one $Z$-partition system, and Algorithm~\ref{alg:alg3} constructs efficiently all the $Z$-partition systems and returns the minimal one.
In the following sections we will introduce our second algorithm which is based on trees. It computes efficiently the reduced $Z(\A_n)$ systems.
%===================================================================================================================================
\section{Tree based reduction}\label{TreeReduction}
%===================================================================================================================================
%\begin{wrapfigure}{r}{0.25\textwidth}
%\centering
%\scalebox{1}{%
%\pstree[levelsep=0.5cm]{\Tcircle{}~[tnpos=a]{The unique $1$-tree}}
%{
%}}
%\scalebox{1}{%
%\pstree[levelsep=0.5cm]{\Tcircle{}~[tnpos=a]{The unique $2$-tree}}
%{
% \Tcircle{}
% \Tcircle{}
%}}
%%\end{center}
%\end{wrapfigure}
%Let us recall some definition on trees. A tree is either empty (no nodes), or a root and zero or more subtrees. Binary trees are best described recursively.
A \emph{binary tree} is a structure defined on a finite set of nodes that either contains no nodes, or is made of three disjoint sets of nodes:
\begin{itemize}
\item a root node
\item a binary tree called its left subtree
\item a binary tree called its right subtree.
\end{itemize}
The binary tree that contains no nodes is called the empty tree. If the left subtree is non-empty, its root is called the left child of the root of the entire tree. Likewise, the root of a non-empty right subtree is the right child of the root of the entire tree. Therefore, in a \emph{full binary tree} each node is either a leaf or has degree exactly $2$, there is no degree-1 nodes. In the following we call a $n$-tree a full binary tree with $n$ leaves. There is a unique $n$-tree for $n=0$ to $2$.
%====================================================================================================================================

Let $t$ be a $n$-tree and let $\pi$ be a path in $t$. The \textit{left weight} (resp. \textit{right weight}) $a_\pi$ (resp. $b_\pi$) is defined as the number of left (resp. right) edges in the path $\pi$. The \textit{length} of $\pi$ denoted by $l_\pi=a_\pi + b_\pi$ is the length of the path $\pi$. Denote by $w_\pi=a_\pi b_\pi$ the \textit{weight} of $\pi$. The \textit{cost} $c_\pi$ of $\pi$ is the sum of its weight and its length. So we have $c_\pi=w_\pi + l_\pi$.

Let $\nu$ be a node in $t$. Denote by $\pi_\nu$ the path from the node $\nu$ to the root of $t$. Denote by $\nu_l$ (resp. $\nu_r$) the left child of $\nu$ (resp. the right child of $\nu$). Denote by $f_\nu$ the father\footnote{The first ancestor.} of $\nu$. If $\pi$ is a path from the node $\nu$ to the root of $t$ then we denote by $f_\pi$ the path from the node $f_\nu$ to the root of $t$. We also associate $a_{\pi_\nu}$, $b_{\pi_\nu}$, $w_{\pi_\nu}$, $l_{\pi_\nu}$ and $c_{\pi_\nu}$ to the node $\nu$ and we denote them respectively by $a_\nu$, $b_\nu$, $w_\nu$, $l_\nu$ and $c_\nu$. The set of leaves of a tree $t$ will be denoted by $L_t$. The weight $w(t)$ of the tree $t$ is defined as the sum of the weight of its leaves, that is $\displaystyle w(t)=\sum_{\nu\in L_t} w_\nu$.
\begin{proposition}
Each $Z(\A_n)$ partition system corresponds to a unique $n$-tree.
\label{propo4}
\end{proposition}
\begin{proof}
The idea is that if we follow the execution trace of the recursive Algorithm~\ref{alg:alg2} we can see that it corresponds to a binary tree whose weight is the number of transitions of the reduced automaton. And by induction we can prove that for each state $q$ we have $a_{\nu_q}=a(q)$ and $b_{\nu_q}=b(q)$. See Figure~\ref{figtree}.
%We can also  prove it by a counting argument such that the number of possible decompositions is the same as the number of binary trees.
\end{proof}
\begin{figure}[H]
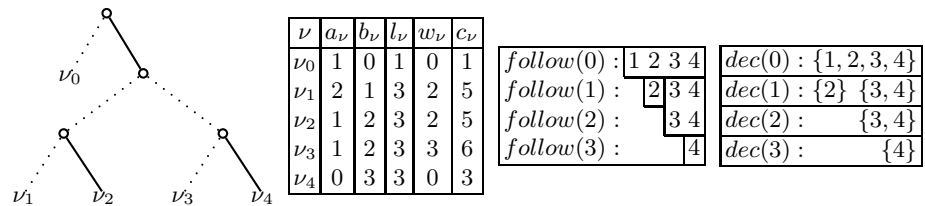

\begin{minipage}{0.30\textwidth}
\begin{center}
\scalebox{1}{%
\pstree[levelsep=.8cm]{\Tc{2pt}}
{
	\TR*[edge=\dedge]{$\nu_0$}^{}
	\pstree{\Tc{2pt}}
	{
		\pstree{\Tc[edge=\dedge]{2pt}}
		{
			\TR*[edge=\dedge]{$\nu_1$}^{}
			\TR*{$\nu_2$}^{}
		}

		\pstree{\Tc[edge=\dedge]{2pt}}
		{
			\TR*[edge=\dedge]{$\nu_3$}^{}
			\TR*{$\nu_4$}^{}
		}
	}
}
}
\end{center}
\end{minipage}
\begin{minipage}{0.22\textwidth}
$\begin{array}{|c|c|c|c|c|c|}\hline
\nu & a_\nu & b_\nu & l_\nu & w_\nu   & c_\nu\\ \hline
\nu_0& 1 & 0 & 1 & 0  & 1 \\
\nu_1& 2 & 1 & 3 & 2 & 5 \\
\nu_2& 1 & 2 & 3 & 2 & 5\\
\nu_3& 1 & 2 & 3 & 3 & 6\\
\nu_4& 0 & 3 & 3 & 0 & 3\\
\hline
\end{array}
$
\end{minipage}
\begin{minipage}{0.48\textwidth}
\begin{tabular}{cccc}
\begin{tabular}{|ccc|cc|}
\hline
\multicolumn{1}{|c|}{$follow(0):$} & $1$ & \multicolumn{1}{c}{$2$} & $3$ & $4$\\
\cline{2-5}
$follow(1):$ & \multicolumn{1}{c|}{} & $2$ & $3$ & $4$\\
\cline{3-3}
$follow(2):$ & & & $3$ & $4$\\
\cline{4-5}
$follow(3):$ & & \multicolumn{1}{c}{ } & \multicolumn{1}{c|}{ } & $4$\\
\hline
\end{tabular} & & \begin{tabular}{|cc|cc|c|}
\hline
$dec(0):$ & \multicolumn{4}{r|}{$\{1,2,3,4\}$}\\
\hline
$dec(1):$ & \multicolumn{2}{r}{$\{2\}$} & \multicolumn{2}{r|}{$\{3,4\}$}\\
\hline
$dec(2):$ & \multicolumn{2}{r}{ } & \multicolumn{2}{r|}{$\{3,4\}$}\\
\hline
$dec(3):$ & \multicolumn{2}{r}{ } & \multicolumn{2}{r|}{$\{4\}$}\\
\hline
\end{tabular}
\end{tabular}
\end{minipage}
\caption{A full binary $5$-tree and an associated $Z(\A_5)$ partition (left edges are represented by dotted lines and right edges with solid lines).}
\label{figtree}
\end{figure}
So, finding a minimal $Z(\A_n)$ partition system is reduced to finding a $n$-tree having minimal weight. Let us denote it by $Z$-tree of rank $n$.

Let $Split(t)$ be the function that returns the tree obtained from $t$ by replacing a leaf having minimal cost in $t$ by the unique $2$-tree. See Figure~\ref{SplitSplitAll}.
\begin{proposition}
The set of $Z$-trees can be computed inductively as follows:
\begin{itemize}
\item $1$-tree is the  $Z$-tree of rank one
\item if $t$ is a $Z$-tree (of rank $i$) then $Split(t)$ is a $Z$-tree (of rank $i+1$).
\end{itemize}
\end{proposition}
\begin{proof}
Let $t_n$ be a $Z$-tree of rank $n$. To get a tree $t_{n+1}$ of rank $n+1$ from $t_n$ we have to split a leaf $\mu$. The weight of $t_{n+1}$ is:
\begin{eqnarray*}
\displaystyle
w(t_{n+1})&=&\sum_{\nu\in L_{t_{n+1}}} w_{\nu}\\
					&=&(\sum_{\nu\in L_{t_{n}}} w_{\nu}) - w_{\mu} + w_{left-child(\mu)}+ w_{right-child(\mu)}\\
					&=&w(t_{n}) - a_{\mu}b_{\mu} + (a_{\mu}+1)b_{\mu} + a_{\mu}(b_{\mu}+1)\\
					&=&w(t_{n})+c_{\mu}
\end{eqnarray*}
If $\mu$ is the leaf of $t_n$ which have the minimal coast, then, the tree $t_{n+1}$ will have  minimal weight.
\end{proof}
%A $Z$-tree $t$ of rank $n$ is called minimal if and only if $w(t)=\min\{w(x)\;|\; x \mbox{ is a $Z$-tree $t$ of rank $n$}\}$.
So, this inductive construction allows us to have the minimal weight tree. The difference of weights between two consecutive minimal trees is exactly the cost of the split leaf. All $Z$-trees of rank less than $n$, can be generated by the following Algorithm~\ref{alg:alg4}.
\begin{algorithm}[H]
\caption{MinZtree $(n)$}
\label{alg:alg4}
\begin{algorithmic}[1]
\REQUIRE $ n\in \mathbb{N}$
\ENSURE $Z$-tree of rank less than $n$
\STATE $t \leftarrow \mbox{$1$-tree}$
\FOR{$i$ = $1$ to $n$}
\STATE $t \leftarrow Split(t)$
\ENDFOR
\end{algorithmic}
\end{algorithm}
\begin{theorem}
Algorithm~\ref{alg:alg4} computes one $Z$-tree of rank $i$ for all $i=1 \mbox{ to } n$ in $O(n \log n)$ time.
\end{theorem}
\begin{proof}
At each step of this algorithm we look for a minimal cost leaf and then we split it. We can maintain the costs of the leaves in a dynamic structure which allow us a logarithmic time search for the minimal cost leaf and also a logarithmic time insertion of the two leaves obtained from the split function.
\end{proof}
It is clear that for a given $n$, there may exist several $Z$-trees of rank $n$.

In the following we introduce a subclass of full binary trees (called $P$-trees), for which the $Z$-trees are unique. We do that in order to study the size-complexity of the reduced automata (the number of transitions).
%We show in addition that this class includes the prime numbers and allows us to generate prime numbers using the function SplitAll.
%===============================================================================
\subsection{$P$-Trees}\label{B}
%===================================================================================================================================
We denote by $M_t$ the set of leaves having minimal cost in $t$, that is: $\displaystyle M_t=\arg \min_{\nu\in L_t} c_\nu$. The function $SplitAll(t)$ returns the tree obtained from $t$ by replacing every leaf in $M_t$ by the unique $2$-tree. See Figure~\ref{SplitSplitAll}.
%====================================================================================================================================
\begin{figure}[H]
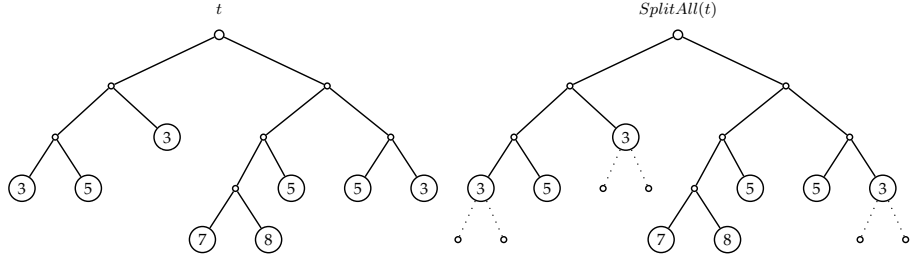

\centering
\resizebox{\textwidth}{!}{%
 \pstree[levelsep=1cm]{\Tc{3pt}~[tnpos=a]{$t$}}
%\pstree[levelsep=1cm]{\Tc{2pt}}
 {
 \pstree{\Tc{2pt}}
 {
 	\pstree{\Tc{2pt}}
		{
	 	\pstree[linecolor=white]{\Tcircle{3}}
			{
			 \Tc{2pt}
			 \Tc{2pt}
			}
		 \Tcircle{5}
		}
		\pstree[linecolor=white]{\Tcircle{3}}
		{
		 \Tc{2pt}
		 \Tc{2pt}
		}
	}
 \pstree{\Tc{2pt}}
 {
 	\pstree{\Tc{2pt}}
		{
  	\pstree{\Tc{2pt}}
	 	{
			 \Tcircle{7}
			 \Tcircle{8}
		 }
		 \Tcircle{5}
		}
		\pstree{\Tc{2pt}}
		{
 	 \Tcircle{5}
	 	\pstree[linecolor=white]{\Tcircle{3}}
			{
			 \Tc{2pt}
			 \Tc{2pt}
			}
		}
	}
 }
% \pstree[levelsep=1cm]{\Tc{3pt}~[tnpos=a]{$Split(t)$}}
% %\pstree[levelsep=1cm]{\Tc{2pt} }
% {
% \pstree{\Tc{2pt}}
% {
% 	\pstree{\Tc{2pt}}
%		{
%	 	%\pstree[edge=\dedge]
%	 	%{\Tcircle{3}}
%		%	{
%		%	 \Tc{2pt}
%		%	 \Tc{2pt}
%		%	}
% \Tcircle{3}		
% \Tcircle{5}
%		}
%		\pstree[edge=\dedge]
%		{\Tcircle{3}}
%		{
%		 \Tc{2pt}
%		 \Tc{2pt}
%		}
%	}
% \pstree{\Tc{2pt}}
% {
% 	\pstree{\Tc{2pt}}
%		{
%  	\pstree{\Tc{2pt}}
%	 	{
%			 \Tcircle{7}
%			 \Tcircle{8}
%		 }
%		 \Tcircle{5}
%		}
%		\pstree{\Tc{2pt}}
%		{
% 	 \Tcircle{5}
%	\Tcircle{3}		
%	
%		}
%	}
% }

 \pstree[levelsep=1cm]{\Tc{3pt}~[tnpos=a]{$SplitAll(t)$}}
 %\pstree[levelsep=1cm]{\Tc{2pt} }
 {
 \pstree{\Tc{2pt}}
 {
 	\pstree{\Tc{2pt}}
		{
	 	\pstree[edge=\dedge]
	 	{\Tcircle{3}}
			{
			 \Tc{2pt}
			 \Tc{2pt}
			}
		 \Tcircle{5}
		}
		\pstree[edge=\dedge]
		{\Tcircle{3}}
		{
		 \Tc{2pt}
		 \Tc{2pt}
		}
	}
 \pstree{\Tc{2pt}}
 {
 	\pstree{\Tc{2pt}}
		{
  	\pstree{\Tc{2pt}}
	 	{
			 \Tcircle{7}
			 \Tcircle{8}
		 }
		 \Tcircle{5}
		}
		\pstree{\Tc{2pt}}
		{
 	 \Tcircle{5}
	 	\pstree[edge=\dedge]
	 	{\Tcircle{3}}
			{
			 \Tc{2pt}
			 \Tc{2pt}
			}
		}
	}
 }
}
	\caption{A $8$-tree $t$ with a $Split$ tree and its $SplitAll$ tree. Values in the nodes are costs.}
	\label{SplitSplitAll}
\end{figure}
%\begin{figure}[H]
%	\centering
%\resizebox{\textwidth}{!}{%
% \pstree[levelsep=1cm]{\Tc{3pt}~[tnpos=a]{$t$}}
%%\pstree[levelsep=1cm]{\Tc{2pt}}
% {
% \pstree{\Tc{2pt}}
% {
% 	\pstree{\Tc{2pt}}
%		{
%	 	\pstree[linecolor=white]{\Tcircle{3}}
%			{
%			 \Tc{2pt}
%			 \Tc{2pt}
%			}
%		 \Tcircle{5}
%		}
%		\pstree[linecolor=white]{\Tcircle{3}}
%		{
%		 \Tc{2pt}
%		 \Tc{2pt}
%		}
%	}
% \pstree{\Tc{2pt}}
% {
% 	\pstree{\Tc{2pt}}
%		{
%  	\pstree{\Tc{2pt}}
%	 	{
%			 \Tcircle{7}
%			 \Tcircle{8}
%		 }
%		 \Tcircle{5}
%		}
%		\pstree{\Tc{2pt}}
%		{
% 	 \Tcircle{5}
%	 	\pstree[linecolor=white]{\Tcircle{3}}
%			{
%			 \Tc{2pt}
%			 \Tc{2pt}
%			}
%		}
%	}
% }
% \pstree[levelsep=1cm]{\Tc{3pt}~[tnpos=a]{$SplitAll(t)$}}
% %\pstree[levelsep=1cm]{\Tc{2pt} }
% {
% \pstree{\Tc{2pt}}
% {
% 	\pstree{\Tc{2pt}}
%		{
%	 	\pstree[edge=\dedge]
%	 	{\Tcircle{3}}
%			{
%			 \Tc{2pt}
%			 \Tc{2pt}
%			}
%		 \Tcircle{5}
%		}
%		\pstree[edge=\dedge]
%		{\Tcircle{3}}
%		{
%		 \Tc{2pt}
%		 \Tc{2pt}
%		}
%	}
% \pstree{\Tc{2pt}}
% {
% 	\pstree{\Tc{2pt}}
%		{
%  	\pstree{\Tc{2pt}}
%	 	{
%			 \Tcircle{7}
%			 \Tcircle{8}
%		 }
%		 \Tcircle{5}
%		}
%		\pstree{\Tc{2pt}}
%		{
% 	 \Tcircle{5}
%	 	\pstree[edge=\dedge]
%	 	{\Tcircle{3}}
%			{
%			 \Tc{2pt}
%			 \Tc{2pt}
%			}
%		}
%	}
% }
%}
%	\caption{A $8$-tree with its SplitAll tree. Numbers in the nodes represent the costs.}
%	\label{SplitAll}
%\end{figure}
%====================================================================================================================================
\begin{definition}
The class $(t_n)_{n>0}$ of $P$-trees is defined inductively as follows:
\begin{itemize}
\item $t_1=1$-tree and
\item $t_{(n+1)}=SplitAll(t_n)$.
\end{itemize} See Figure~\ref{Ptrees}.
\end{definition}
\begin{remark}
Notice that if $\nu \in M_{t_n}$ then $c_\nu=(n-1)$. This can be established by induction on $n$.  Therefore, to get $t_{(n+1)}$ from $t_n$ we split leaves of cost $(n-1)$.
\label{remark1}
\end{remark}
\newgray{lightgray}{.75}
\begin{figure}[H]
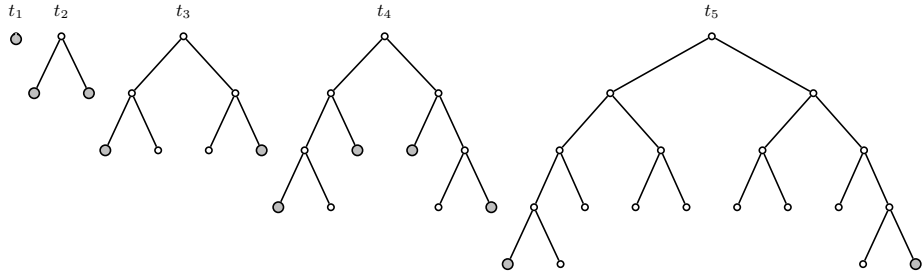

	\centering
\resizebox{1.00\textwidth}{!}{%
\pstree[levelsep=.05cm,linecolor=lightgray,fillstyle=solid]{\Tc{2pt}~[tnpos=a]{$t_1$}}
{
 \Tc[linecolor=black,fillcolor=lightgray,fillstyle=solid]{3pt}
}

%\pstree[levelsep=1cm]{\Tc{3pt}~[tnpos=a]{$t_1$}}
\pstree[levelsep=1cm]{\Tc{2pt}~[tnpos=a]{$t_2$}}
{
			\Tc[fillcolor=lightgray,fillstyle=solid]{3pt}
			\Tc[fillcolor=lightgray,fillstyle=solid]{3pt}
}

\pstree[levelsep=1cm]{\Tc{2pt}~[tnpos=a]{$t_3$}}
{
	\pstree{\Tc{2pt}}
	{
			\Tc[fillcolor=lightgray,fillstyle=solid]{3pt}
			\Tc{2pt}
	}
	\pstree{\Tc{2pt}}
	{
			\Tc{2pt}
			\Tc[fillcolor=lightgray,fillstyle=solid]{3pt}
	}
}

\pstree[levelsep=1cm]{\Tc{2pt}~[tnpos=a]{$t_4$}}
{
	\pstree{\Tc{2pt}}
	{
			\pstree{\Tc{2pt}}
			{
				\Tc[fillcolor=lightgray,fillstyle=solid]{3pt}
				\Tc{2pt}
			}
			\Tc[fillcolor=lightgray,fillstyle=solid]{3pt}
	}
	\pstree{\Tc{2pt}}
	{
			\Tc[fillcolor=lightgray,fillstyle=solid]{3pt}
			\pstree{\Tc{2pt}}
			{
				\Tc{2pt}
				\Tc[fillcolor=lightgray,fillstyle=solid]{3pt}
			}
	}
}

\pstree[levelsep=1cm]{\Tc{2pt}~[tnpos=a]{$t_5$}}
{
	\pstree{\Tc{2pt}}
	{
			\pstree{\Tc{2pt}}
			{
				\pstree{\Tc{2pt}}
				{
					\Tc[fillcolor=lightgray,fillstyle=solid]{3pt}
					\Tc{2pt}
				}
				\Tc{2pt}
			}
							\pstree{\Tc{2pt}}
				{
					\Tc{2pt}
					\Tc{2pt}
				}
	}
	\pstree{\Tc{2pt}}
	{
							\pstree{\Tc{2pt}}
				{
					\Tc{2pt}
					\Tc{2pt}
				}
			\pstree{\Tc{2pt}}
			{
				\Tc{2pt}
				\pstree{\Tc{2pt}}
				{
					\Tc{2pt}
					\Tc[fillcolor=lightgray,fillstyle=solid]{3pt}
				}
			}
	}
}
}
	\caption{The first four $P$-trees. Leaves with big circle have minimal cost.}
	\label{Ptrees}
\end{figure}
\subsection{Eratosthenes-Pascal's Triangle}
%====================================================================================================================================
The Eratosthenes-Pascal's Triangle is constructed from Pascal's Triangle as follows: we interleave each column $k$ of Pascal's Triangle with $(k-1)$ zeros. The element of the Eratosthenes-Pascal's Triangle at the $n^{th}$ row  and the $k^{th}$ column is denoted by $T_n^k$ with $n\geq 1$ and $k\geq 1$.
%====================================================================================================================================
\begin{center}
\begin{tabularx}{.5\textwidth}{l|XXXXXXXXXXX}
%\begin{figure}[H]
%\centering
%$\begin{array}{l|c|c|c|c|c|c|c|c|c|c|c|}
&1&2&3&4&5&6&7&8&9&10&$\cdots$\\\hline
1&1&&&&&&&&&&\\
\textbf{2}&1&1&&&&&&&&&\\
\textbf{3}&1&0&1&&&&&&&&\\
4&1&2&0&1&&&&&&&\\
\textbf{5}&1&0&0&0&1&&&&&&\\
%6&\textbf{1}&\textbf{3}&\textbf{3}&\textbf{0}&\textbf{0}&\textbf{1}&&&&&\\
6&1&3&3&0&0&1&&&&&\\
\textbf{7}&1&0&0&0&0&0&1&&&&\\
8&1&4&0&4&0&0&0&1&&&\\
9&1&0&6&0&0&0&0&0&1&&\\
10&1&5&0&0&5&0&0&0&0&1&\\
\textbf{11}&1&0&0&0&0&0&0&0&0&0&1\\
%\end{array}$
%\caption{Eratosthenes-Pascal's Triangle}
%\label{f4}
%\end{figure}
\end{tabularx}
\end{center}
%====================================================================================================================================
Let $n$ be a natural number. We denote by $D_n$ the set of divisors of $n$.
%====================================================================================================================================
\begin{proposition}
\begin {eqnarray*}
 T_n^k&=&
 \left\{
 \begin{array}{ll}
  \displaystyle \binom{\displaystyle\frac{n}{k}+k-2}{k-1} & \mbox{ if } k \in D_n\\
  0 & \mbox{ otherwise. }
\end{array} \right.
\end {eqnarray*}
\end{proposition}
%====================================================================================================================================
\begin{proof}
The element of the Pascal's Triangle at the $n^{th}$ row and the $j^{th}$ column is $\binom{i-1}{j-1}$. This element is moved in the Eratosthenes-Pascal's Triangle to the row $r=(i-j+1)j$ in the same column $j$. We have then
\begin{eqnarray*}
T_{(i-j+1)j}^j=\binom{i-1}{j-1} &\mbox{ Thus }& T_r^j=\binom{\displaystyle\frac{r}{j}+j-2}{j-1}
\end{eqnarray*}
%\qedhere
\end{proof}
%====================================================================================================================================
Let $S_n$ be the sum of the $n^{th}$ row in the Eratosthenes-Pascal's Triangle.
\begin{eqnarray*}
S_n&=&\displaystyle\sum_{k = 1}^{n}T_n^k=\sum_{k \in D_n}{}\binom{ \displaystyle\frac{n}{k}+k-2}{k-1}
\end{eqnarray*}
%%====================================================================================================================================
%\begin{corollary}\label{prime2}
%Let $n$ be a natural number. Then $n$ is prime if and only if $S_n=2$.
%\end{corollary}
%%====================================================================================================================================
%\begin{proof}
%A number $n$ is a prime if and only if the $n^{th}$ row in the Eratosthenes-Pascal's Triangle is
%%====================================================================================================================================
%\begin{eqnarray*}
%\overbrace{1,\underbrace{0,0,...,0}_{zeros},1}^{n \; elements}
%\end{eqnarray*}
%%\qedhere
%\end{proof}
%====================================================================================================================================
\subsection{$P$-trees and Eratosthenes-Pascal's Triangle}
%====================================================================================================================================
In this section we will describe the link between the Eratosthenes-Pascal's Triangle and the set of $P$-trees.
%====================================================================================================================================
\begin{theorem}
%$P$-trees are in bijection with Eratosthenes-Pascal's Triangle's rows, and
The sum of the elements of the $n^{th}$ row of Eratosthenes-Pascal's Triangle's, $S_n$, is exactly $|M_{t_n}|$ the number of leaves of minimal cost in the $P$-tree $t_n$.
\end{theorem}
%====================================================================================================================================
To prove this theorem we introduce a family $F_n$ which is in bijection with both the Eratosthenes-Pascal's Triangle's rows and with $P$-trees.
%====================================================================================================================================
We focus on the following question:
\emph{Given a natural number $n$ what are all the possible paths that have $(n-1)$ as cost?}
%====================================================================================================================================
%how we can determine all possible paths that are associated with a leaf $\nu$ of cost $c_\nu=n$?
To answer this question, we define $F_{(n-1)}$ as the set of all paths of cost $(n-1)$:
%====================================================================================================================================
\begin{eqnarray*}
F_{(n-1)}&=&\{\pi \;\mid\; c_\pi=(n-1) \}
\end{eqnarray*}
%====================================================================================================================================
\begin{lemma}
%$F_n$ family is in bijection with Eratosthenes-Pascal's Triangle's rows.
For all $n\geq 1$, $|F_{(n-1)}|=S_n$.
\end{lemma}
%====================================================================================================================================
\begin{proof}
Let $F_{(n-1)}^i$ for $0 \leq i \leq (n-1)$ be the set of paths of $F_{(n-1)}$ having $i$ left edges. We have $\displaystyle F_{(n-1)}=\bigcup^{(n-1)}_{i=0} F_{(n-1)}^i$ where
%====================================================================================================================================
\begin{eqnarray*}
F_{(n-1)}^i&=&\{\pi \in F_{(n-1)} \;\mid\; a_\pi=i \}\\
   &=&\{\pi \;\mid\; (a_\pi b_\pi+a_\pi+b_\pi=(n-1))\wedge(a_\pi=i) \}\\
   &=&\{\pi \;\mid\; (b_\pi=\displaystyle\frac{n}{i+1}-1) \wedge (a_\pi=i) \}
\end{eqnarray*}
%====================================================================================================================================
%====================================================================================================================================
\begin {eqnarray*}
 \mbox{So, we get}\;\;|F_{(n-1)}^i|&=&
 \left\{
 \begin{array}{ll}
  \displaystyle \binom{\displaystyle\frac{n}{i+1}+i-1}{i} & \mbox{ if } (i+1) \in D_{n}\\
  0& \mbox{ otherwise. }\\
\end{array} \right.
\end {eqnarray*}
%====================================================================================================================================
This corresponds to the different ways to arrange $i$ left edges in a path of length $\displaystyle\frac{n}{i+1}+i-1$. Let $k=i+1$ then
%====================================================================================================================================
\begin {eqnarray*}
 |F_{(n-1)}^{k-1}|&=&
 \left\{
 \begin{array}{ll}
  \displaystyle \binom{\displaystyle\frac{n}{k}+k-2}{k-1} & \mbox{ if } k \in D_{n}\\
  0 & \mbox{ otherwise. }\\
\end{array} \right.
\end {eqnarray*}
%====================================================================================================================================
Thus $|F_{(n-1)}^{k-1}|=T_{n}^{k}$.
%====================================================================================================================================
That is for $0 \leq i < n$ we get $|F_{(n-1)}^{i}|=T_{n}^{i+1}$.
%====================================================================================================================================
Finally
%====================================================================================================================================
\begin {eqnarray*}
|F_{(n-1)}|&=&\sum^{(n-1)}_{i=0} |F_{(n-1)}^{i}|=\sum^{(n-1)}_{i=0} T_{n}^{i+1}=\sum^{n}_{i=1} T_{n}^{i}=S_{n}
\end {eqnarray*}
%====================================================================================================================================
Therefore, we can associate the $n^{th}$ row of the Eratosthenes-Pascal's Triangle to $F_{(n-1)}$.
\end{proof}
%====================================================================================================================================
\newgray{lightgray}{.8}
\begin{figure}[H]
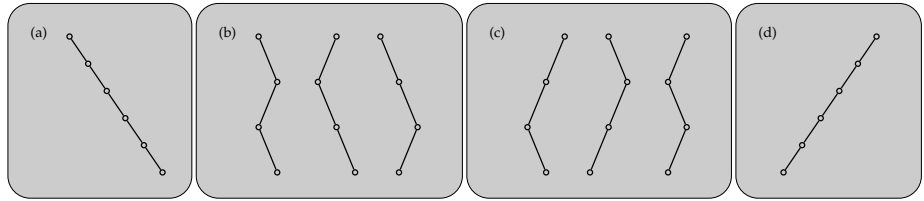

	\centering
\resizebox{1.00\textwidth}{!}{%
\psframebox[fillstyle=solid,fillcolor=lightgray,framesep=14pt,linearc=14pt,cornersize=absolute,linewidth=0.5pt]{%
(a)
\pstree[levelsep=0.6cm]{\Tc{2pt}}
{
	\Tn
	\pstree{\Tc{2pt}}
	{
		\Tn
		\pstree{\Tc{2pt}}
		{
			\Tn
			\pstree{\Tc{2pt}}
			{
				\Tn
				\pstree{\Tc{2pt}}
				{
					\Tn
					\Tc{2pt}
				}
			}
		}
	}
}
}
\psframebox[fillstyle=solid,fillcolor=lightgray,framesep=14pt,linearc=14pt,cornersize=absolute,linewidth=0.5pt]{%
(b)
\pstree[levelsep=1cm]{\Tc{2pt} }
{
	\Tn
	\pstree{\Tc{2pt}}
	{
		\pstree{\Tc{2pt}}
		{
			\Tn
			\Tc{2pt}
		}
		\Tn
	}
}

\pstree[levelsep=1cm]{\Tc{2pt} }
{
	\pstree{\Tc{2pt}}
	{
		\Tn
		\pstree{\Tc{2pt}}
		{
			\Tn
			\Tc{2pt}
		}
		}
	\Tn
}

\pstree[levelsep=1cm]{\Tc{2pt} }
{
	\Tn
	\pstree{\Tc{2pt}}
	{
		\Tn
		\pstree{\Tc{2pt}}
		{
			\Tc{2pt}
			\Tn
		}
		}	
}

}
\psframebox[fillstyle=solid,fillcolor=lightgray,framesep=14pt,linearc=14pt,cornersize=absolute,linewidth=0.5pt]{%
(c)
\pstree[levelsep=1cm]{\Tc{2pt} }
{
	\pstree{\Tc{2pt}}
	{
		\pstree{\Tc{2pt}}
		{
			\Tn
			\Tc{2pt}
		}
			\Tn
		}
	\Tn	
}

\pstree[levelsep=1cm]{\Tc{2pt} }
{
	\Tn
	\pstree{\Tc{2pt}}
	{
		\pstree{\Tc{2pt}}
		{
			\Tc{2pt}
			\Tn
		}
			\Tn
		}
}

\pstree[levelsep=1cm]{\Tc{2pt} }
{
	\pstree{\Tc{2pt}}
	{
		\Tn
		\pstree{\Tc{2pt}}
		{
			\Tc{2pt}
			\Tn
		}
	}
	\Tn
}
}
\psframebox[fillstyle=solid,fillcolor=lightgray,framesep=14pt,linearc=14pt,cornersize=absolute,linewidth=0.5pt]{%
(d)
\pstree[levelsep=0.6cm]{\Tc{2pt} }
{
	\pstree{\Tc{2pt}}
	{
		\pstree{\Tc{2pt}}
		{
			\pstree{\Tc{2pt}}
			{
				\pstree{\Tc{2pt}}
				{
					\Tc{2pt}
					\Tn
				}
				\Tn
			}
			\Tn
		}
		\Tn
	}
	\Tn
}
}
}
	\caption{All paths of cost $5$.
	(a) The unique path of cost $5$ with no  left edges which corresponds to $\binom{5}{0}=1$.
	(b) The three paths of cost $5$ with one left edges which corresponds to $\binom{3}{1}=3$.
	(c) The three paths of cost $5$ with two left edges which corresponds to $\binom{3}{2}=3$.
	(d) The unique path of cost $5$ with five left edges which corresponds to $\binom{5}{5}=1$.
	There is no paths of cost $5$ with three or four left edges that correspond to zeros in the $6^{th}$ row of the Eratosthenes-Pascal's Triangle.
	}
	\label{f5}
\end{figure}
%====================================================================================================================================
\begin{lemma}
%The family $(F_n)_{n\geq 0}$, is in bijection with $P$-trees.
For all $n\geq 1$, $|F_{(n-1)}|=|M_{t_n}|$.
\end{lemma}
%====================================================================================================================================
\begin{proof}
%====================================================================================================================================
We will show that each set $F_n$ is associated with the $P$-tree $t_n$. We proceed by induction. Assume that for all $k \leq n$, the $P$-tree $t_n$
contains all paths of cost less or equal to $n$, that is, for all $k \leq n$, $F_k$  is in the $P$-tree $t_n$.
%====================================================================================================================================
From this hypothesis we should show that $t_{(n+1)}=SplitAll(t_n)$ contains all the paths of cost less or equal to $(n+1)$.
Of course, $t_{(n+1)}$ contains all the paths of cost less or equal to $n$ because the tree $t_{(n+1)}$ is obtained from $t_n$. However, \emph{does $t_{(n+1)}$ contain all paths of cost equal to $(n+1)$?} To answer this question, we proceed by absurd. Let $\pi$ be a path with cost $c_\pi=(n+1)$ which is not in $t_{(n+1)}$. It is clear that the cost of the father of $\pi$ is:
%Assume that $t_{(n+1)}$ does not contain all paths of weight $(n+1)$,
%====================================================================================================================================
\begin {eqnarray*}
 \displaystyle
 c_{f_\pi}&=&
 \left\{
 \begin{array}{ll}
  (a_\pi-1) b_\pi+(a_\pi-1)+b_\pi & \mbox{ if $\pi$ is a left child},\\
  a_\pi (b_\pi-1)+a_\pi+(b_\pi-1) & \mbox{ if $\pi$ is a right child}.
\end{array} \right.
\end{eqnarray*}
%====================================================================================================================================
As $c_{f_\pi} < c_\pi$ then $c_{f_\pi} \leq n$. From this we deduce that when $c_{f_\pi} < n$, the father $f_\pi$ of $\pi$ was split by the $SplitAll()$ function within an earlier or it will be split in the current tree $t_n$ in the case where $c_{f_\pi} = n$ (see Remark~\ref{remark1}). In both cases, the path $\pi$ is necessarily in $t_{(n+1)}$.
%For that, we can associate $M_{t_n}$ the set of minimal leaves in $t_n$ (also the set of leaves of cost $(n-1)$) to $F_{(n-1)}$.
%====================================================================================================================================
\end{proof}
%====================================================================================================================================
From the two last lemmas, we construct a bijection between the $P$-trees and the Eratosthenes-Pascal's Triangle's rows, and we claim following corollary:
%====================================================================================================================================
\begin{corollary}\label{Sdefn}
\label{cor1}
Let $t_n$ be a $P$-tree. Then $|M_{t_n}|=|F_{(n-1)}|=S_n$.
\end{corollary}
%====================================================================================================================================
From Corollary~\ref{cor1}, we have $\displaystyle S_n=\sum_{k|n}{}\binom{ \frac{n}{k}+k-2}{k-1}$.
\begin{proposition}\label{wsdefn}
Let $t_n$ be a $P$-tree. Its size $s(t_n)=|L_{t_n}|$ (the number of leaves), and its weight $w(t_n)$ are: $\displaystyle s(t_n)=1+\sum_{i=1}^{n-1}S_{i}$ and $\displaystyle w(t_n) =\sum_{i=2}^{n}(i-2)S_{i-1}$.%With $s(t_1)=w(t_1)=0$.
\end{proposition}
\begin{proof}
From Corollary~\ref{cor1} and Remark~\ref{remark1}, we have, the size of the tree $t_n$ is the sum of the size of the tree $t_{n-1}$ and the number of all split minimal leaves, that is, $s(t_n)=s(t_{n-1})+|M_{t_{n-1}}|$. We have also, the weight of the tree $t_n$ is the sum of the weight of the tree $t_{n-1}$ and the costs of all split minimal leaves, that is, $w(t_n)=w(t_{n-1})+(n-2)|M_{t_{n-1}}|$.
\end{proof}
From Proposition~\ref{propo4}, the $P$-tree $t_n$ (which is a $Z$-tree) corresponds to a $Z(\A_{s(t_n)})$ partition system. The CFS automaton associated with $Z(\A_{s(t_n)})$ has $w(t_n)$ transitions.

With our construction the CFS automaton associated with a $Z(\A_n)$ CFS partition system may contain several initial states. However, in order to compare the number of minimal automata and their number of transitions, with those obtained by R. Cox \cite{Cox} (seen Table~\ref{tab:t1}), we must restrict our study to CFS partition systems leading to a unique initial state automata. It is easy to verify that in this case a $P$-tree $t_n$ has $\displaystyle s(t_n) =1+\sum_{i=1}^{n-1} \frac{S_{i+2}}{2}$ and $\displaystyle w(t_n) =\sum_{i=1}^{n} \frac{iS_{i+1}}{2}$.
%===============================================================================
\begin{center}
\begin{table}[H]
\begin{tabularx}{\linewidth}{|l|X|X|X|X|X|X|X|X|X|X|X|X|X|X|}\hline
\textbf{n}&	$1$&	$2$&	$3$&	$4$&	$5$&	$6$&	$7$&	$8$&	$9$&	$10$&	$11$&	$12$&	$13$&	$14$\\\hline
(i)&	$1$&	$3$&	$6$&	$9$&	$13$&	$18$&	$23$&	$\leq 28$&	$\leq 34$&	$\leq 41$&	$?$&	$?$&	$?$&	 $?$\\\hline
(ii)&	$1$&	$1$&	$2$&	$1$&	$1$&	$4$&	$6$&	$\geq 1$&	$\geq 1$&	$\geq 1$&	$?$&	$?$&	$?$&	 $?$\\\hline
(iii)&	$1$&	$3$&	$6$&	$9$&	$13$&	$18$&	$23$&	$28$&	$33$&	$39$&	$46$&	$53$&	$60$&	$67$\\\hline
(iv)&	$1$&	$1$&	$2$&	$1$&	$1$&	$4$&	$6$&	$4$&	$1$&	$1$&	$5$&	$10$&	$10$&	$5$\\\hline
\end{tabularx}
\caption[LiALLAHalhamd]
{Comparison table.
\begin{inparaenum}[(i)]
\item Minimal transition number estimated by R. Cox \cite{Cox}
\item Number of minimal automata estimated by R. Cox \cite{Cox}
\item Number of transitions in our reduced automaton
\item Number of reduced automata estimated by our approach.
\end{inparaenum}}
\label{tab:t1}
\end{table}
\end{center}
%%===============================================================================
%Using generating function we have computed the number of transitions of the automata associated with minimal $Z(\A_n)$ partition system up to billions of states.
%%===============================================================================
%\begin{figure}[H]
%	\centering
%	\includegraphics[width=\textwidth]{plot4.eps}
%	\caption{Experimental results.}%PARI/GP\cite{PARI2} and GnuPlot\cite{Gnuplot}.}
%	\label{fig:plot1}
%\end{figure}
%%===============================================================================
%%we have calculate the sequence of $w(t_n)$ for trees of some billion leafs. See Figure~\ref{fig:plot1}.
%%========================================================================================================
\section{Asymptote behavior of the number of transitions}\label{Asymptote}
%===================================================================================================================================
Appendix~\ref{AsymptoticNotations} contains the basics of asymptotic notations. In this section we shall establish one of our main result which concern the asymptotically result on the behavior of the number of transitions $w(t_n)$, where $s(t_n)$ is the number of states. Namely, we will show that the weight of our automata is asymptotically equivalent to $s(t_n)log^2s(t_n)$ up to constant which means that the number of transitions is minimal in the sense that we reach the lower bounded of Shnitger \cite{Schnitger06}. Indeed, we have
\begin{theorem}\label{mainth1}$\displaystyle \log^2(4)~w(t_n) \sim s(t_n)\log^2 s(t_n).$
\end{theorem}
 As a consequence we obtain the following
\begin{corollary}For a large $n$, we have
\[
\omega(t_n) < s(t_n)\log^2 s(t_n).
\]
\end{corollary}
It is also easy to deduce the following
\begin{corollary}\label{mainth2}$\displaystyle \frac{s(t_n)\log^2 s(t_n)}{\log\log~s(t_n)}=o\big(w(t_n)\big).$
\end{corollary}
Before starting the proof of Theorem \ref{mainth1}. Observe that according to the Corollary \ref{Sdefn} combined with
Proposition \ref{wsdefn}, one may consider that  $\omega(t_n)$ and
$s(t_n)$ are given by
$$ s(t_n)=1+\sum_{1}^{n-1}S_i {\rm {~~and~~}} \omega(t_n)= \sum_{i=1}^{n-2}i S_{i+1}.$$
As usual in the number theory, any arithmetical function $f~:~\Z \longrightarrow \R$ can be extended to the real line
by putting, for any $ x\in \R$, $f(x)=f(\lfloor x \rfloor)$. Therefore, for any $x \geq 2,$ we have
$$ s(t_x)=1+\sum_{1}^{\lfloor x\rfloor-1}S_i {\rm {~~and~~}} \omega(t_x)= \sum_{i=1}^{\lfloor x \rfloor-2}i S_{i+1}.$$
We recall that the classical arithmetical function $\pi(x)$ denote the number of primes not exceeding $x$.
We shall need also the following classical identity due to Abel
\begin{theorem}[Abel's identity, \cite{Apostol}]For any arithmetical function $a(n)$ let
 $$ A(x)= \sum_{n \leq x} a(n)$$
where $A(x)=0$ if $x <1$. Assume $f$ has a continuous derivative on the interval $[y,x]$, where
$0<y<x$. Then we have
\begin{eqnarray}\label{Abel}
 \sum_{y < n \leq x} a(n) f(n)= A(x)f(x)-A(y)f(y)-\bigintss_{y}^{x}A(u) f'(u) du..
\end{eqnarray}
\end{theorem}
We deduce easily from the Abel's identity the following lemma
\begin{lemma}\label{Abellemma} For any integer $n \geq 1$, we have
$$\omega(t_n)=(n-2)\widetilde{s}(t_n)-3-\bigintss_{2}^{n-2}\widetilde{s}(t_u) du,$$
where $\widetilde{s}(t_u)=s(t_u)-1.$
\end{lemma}
We need to estimate $\omega(t_n)$ with respect to $s(t_n)$. For that, we shall need the following weaker form
of the Prime Number Theorem (WPNT for short) due to Cheyshev
\begin{theorem}[\cite{Apostol}]\label{PNTChe} For every integer $n \geq 2$ we have
\begin{eqnarray}
 \frac16 \frac{n}{\log(n)} <\pi(n) <6 \frac{n}{\log(n)}.
\end{eqnarray}

\end{theorem}
We deduce from the WPNT the following crucial proposition
\begin{proposition}\label{minorer}
For all $u \geq 4$ we have
\begin{eqnarray}\label{wpnt}
 \widetilde{s}(t_u) \geq \frac13 \frac{u-2}{\log(u-2)}.
\end{eqnarray}
\end{proposition}
%\begin{proof}By Corollary \ref{prime2}, for any prime number $p \geq 2$, we have $S_p=2$. Hence
\begin{proof} We can show that for any prime number $p \geq 2$, we have $S_p=2$. Hence
$$\widetilde{s}(t_u)=\sum_{i \leq u}S_i \geq 2 \pi(u-1).$$
Consequently, by WPNT we get
$$\widetilde{s}(t_u) \geq \frac13 \frac{[u]-1}{\log([u]-1)}.$$
But the function $x \in [3, +\infty[ \mapsto \frac{x}{\log(x)}$ is increasing function. It follows that we have,
For all $u \geq 3$,
$$ \widetilde{s}(t_u) \geq \frac13 \frac{u-2}{\log(u-2)},$$
Which achieve the proof of the proposition.
\end{proof}
For any $x \geq 2$ and any positive integer $n$, let
\[
\Li_n(x)=\bigintss_{2}^{x}\frac{dt}{\log^n(t)}.
\]
Let us summarize in the following proposition a classical well-known results on $Li_n(x)$ that we shall used.
\begin{proposition} For every $x \geq 2$  and integer $n \geq 1$, we have
\begin{eqnarray}
\Li_1(x)&=&\frac{x}{\log(x)}+\Li_2(x)-\frac{2}{\log(2)},\\
\Li_n(x)&=&O\big(\frac{x}{\log^n(x)}\big).
\end{eqnarray}
\end{proposition}

Now, we are able to formulate our key estimation of $\omega(t_n)$ with respect to $s(t_n)$ in the following proposition.
\begin{proposition}\label{limsup}
$\displaystyle \limsup\big(\frac{\omega(t_n)}{(n-2)\widetilde{s}(t_n)} \big)\leq 1.$
\end{proposition}
\begin{proof} Applying Lemma \ref{Abellemma} we have, for any $n \geq 3$,
$$ \frac{\omega(t_n)}{(n-2)\widetilde{s}(t_n)} = 1-\frac{3}{(n-2)\widetilde{s}(t_n)}-\frac1{(n-2)\widetilde{s}(t_n)}
\bigintss_{2}^{n-2}\widetilde{s}(t_u) du.$$
Therefore, for any $n \geq 5$, we have
\begin{eqnarray*}
  \frac{\omega(t_n)}{(n-2)\widetilde{s}(t_n)}\leq1-\frac{3}{(n-2)\widetilde{s}(t_n)}
\end{eqnarray*}
From Proposition \ref{minorer} we deduce, that for any $n \geq 5$ we have
$$ \frac{3}{(n-2)\widetilde{s}(t_n)}\leq \frac{1}{\log(n-2)}$$
Hence, by letting $n$ goes to $\infty$  we obtain
$$
\frac3{(n-2)\widetilde{s}(t_n)}  \tendn 0.
$$
Which implies that
$$
\displaystyle \limsup\big(\frac{\omega(t_n)}{(n-2)\widetilde{s}(t_n)} \big)\leq 1,
$$
and this finish the proof of the proposition.
\end{proof}
We are reduced to compare the sequences $(n-2)s(t_n)$ and $s(t_n)\log^2 s(t_n)$. For that
we shall estimate $s(t_x)$. Precisely, we argue that we have the following
\begin{theorem}\label{Catalan}For a large $x>0$ we have
$$x^{\frac34}\frac{4^{\sqrt{x}-1}}{\sqrt{\pi}} \leq s(t_x)\leq x^{\frac34}
\log(x)\frac{4^{\sqrt{x}-1}}{\sqrt{\pi}}.$$
\end{theorem}
The proof of Theorem \ref{Catalan} will be given later. For instance, using Theorem \ref{Catalan} holds we
shall extended Proposition \ref{limsup} as follows.
\begin{proposition}\label{equiva1}The sequences $\omega(t_n)$ and $(n-2)s(t_n)$ two sequences be equivalent. That is,
$$\frac{\omega(t_n)}{(n-2)s(t_n)}\tendn 1.$$
\end{proposition}
\begin{proof} By Lemma \ref{Abellemma}, write
$$\frac{\omega(t_n)}{(n-2)s(t_n)}=1-\frac{3}{(n-2)s(t_n)}-\frac{1}{(n-2)s(t_n)}\bigintss_{2}^{n-2}s(t_x) dx.$$
Let $\varepsilon>0$ and $x$ sufficiently large. Then, by Theorem \ref{Catalan}, for a large $x$, we have
\begin{eqnarray}\label{key-relation}
x^{\frac34}\frac{4^{\sqrt{x}-1}}{\sqrt{\pi}} \leq s(t_x)\leq x^{\frac34}
\log(x)\frac{4^{\sqrt{x}-1}}{\sqrt{\pi}}.
\end{eqnarray}
But
\begin{eqnarray}\label{int4u}
\bigintss_{2}^{n-2} 4^{\sqrt{x}} dx& \overset{u=\sqrt{x}}{=}&
\bigintss_{\sqrt{2}}^{\sqrt{n-2}} 2~u~ 4^{u} du \nonumber\\
&=&\Big[ \frac{2u}{\log(4)}4^u\Big]_{\sqrt{2}}^{\sqrt{n-2}}-\Big[ \frac{2}{\log^2(4)}4^u\Big]_{\sqrt{2}}^{\sqrt{n-2}} \nonumber\\
&=&\frac{2\sqrt{n-2}~.4^{\sqrt{n-2}}}{\log(4)}-\frac{2\sqrt{2}~.4^{\sqrt{2}}}{\log(4)}-
\frac{2~.4^{\sqrt{n-2}}}{\log^2(4)}+\frac{2~.4^{\sqrt{2}}}{\log^2(4)}.
\end{eqnarray}
Since $\frac{1}{(n-2)s(t_n)}$ vanishes at the infinity and $\widetilde{s}(t_n)$ is equivalent to
$s(t_n)$ we may assume that \eqref{key-relation} holds starting from
$2$ and Theorem \ref{Catalan} is valid for $\widetilde{s}(t_n)$. Therefore

\begin{eqnarray}\label{upbound}\displaystyle
\frac{1}{(n-2)\widetilde{s}(t_n)}\bigintss_{2}^{n-2} \widetilde{s}(t_x) dx
&\leq& \frac1{(n-2)(n-2)^{\frac34}~4^{\sqrt{n}}}\bigintss_{2}^{n-2} x^{\frac34}
\log(x)~4^{\sqrt{x}} dx \nonumber\\
&\leq& \frac{(n-2)^{\frac34}\log(n-2)}{(n-2)(n-2)^{\frac34}~4^{\sqrt{n}}}\bigintss_{2}^{n-2} 4^{\sqrt{x}} dx \nonumber\\
&\leq& \frac{\log(n-2)}{(n-2)4^{\sqrt{n}}}\bigintss_{2}^{n-2} 4^{\sqrt{x}} dx.
\end{eqnarray}
From \eqref{int4u} combined with \eqref{upbound} it follows that
\begin{eqnarray*}
&&\frac{1}{(n-2)\widetilde{s}(t_n)}\bigintss_{2}^{n-2} \widetilde{s}(t_x) dx
\leq\\&&\frac{2\sqrt{n-2}~.4^{\sqrt{n-2}}}{\log(4)} \times \frac{\log(n-2)}{(n-2)4^{\sqrt{n}}}+
\frac{2~.4^{\sqrt{2}}}{\log^2(4)}\times \frac{\log(n-2)}{(n-2)4^{\sqrt{n}}} \tendn 0.
\end{eqnarray*}
We conclude that
$$\frac{\omega(t_n)}{(n-2)\widetilde{s}(t_n)}\tendn 1.$$
which proves the proposition.
\end{proof}
It remains to prove Theorem \ref{Catalan}. For that we shall need the following classical lemma. The proof of it
can be found in \cite{Flajolet}. Nevertheless we include the proof for the sake of completeness.
\begin{lemma}\label{Catalan2}
$\displaystyle \left(
  \begin{array}{cc}
    2n \\
    n
  \end{array}
\right)\sim \displaystyle \frac{4^n}{\sqrt{\pi}\sqrt{n}}.$
\end{lemma}
\begin{proof}By Stirling formula we have
\[
n!\sim n^n e^{-n}\sqrt{2\pi n}
\]
Hence
\[
\left(
  \begin{array}{cc}
    2n \\
    n
  \end{array}
\right)=\frac{2n!}{(n!)^2}\sim \frac{4^n}{\sqrt{\pi}\sqrt{n}}.
\]
This finishes the proof of the lemma.
\end{proof}
\begin{proof}[Proof of Theorem \ref{Catalan}]{} For $x \geq 2$, Write
\begin{eqnarray*}\displaystyle
s(t_x)&=& \sum_{n \leq x}\sum_{d|n}\left(
                                    \begin{array}{c}
                                      \frac{n}{d}+d-2 \\
                                      d-1\\
                                    \end{array}
                                  \right)\\
&=&\displaystyle\sum_{dq \leq x}\left(
\begin{array}{c}
                                      q+d-2 \\
                                      d-1\\
                                    \end{array}
                                  \right)\\
&=&\displaystyle\sum_{d=1}^{\lfloor x \rfloor}\sum_{q=1}^{\lfloor \frac{x}{d} \rfloor }
\left(\begin{array}{c}
                                      q+d-2\\
                                      d-1\\
                                    \end{array}
                                  \right)\\
\end{eqnarray*}
From this we see that
\begin{eqnarray}\label{CCatalan}\displaystyle
s(t_x)&\leq& \sum_{d=1}^{\lfloor x \rfloor}\lfloor \frac{x}{d} \rfloor
\Big(\begin{array}{c}
                                      2(\lfloor x \rfloor-1)\\
                                      \lfloor x\rfloor-1\\
                                    \end{array}
                                    \Big) \nonumber\\
&\leq & x \log(x)  \Big(\begin{array}{c}
                                      2(\lfloor x \rfloor-1)\\
                                      \lfloor x\rfloor-1\\
                                    \end{array}
                                    \Big),
\end{eqnarray}
and
\begin{eqnarray}\label{CCCatalan}\displaystyle
s(t_x)\geq \lfloor x\rfloor \Big(\begin{array}{c}
                                      2(\lfloor x \rfloor-1)\\
                                      \lfloor x\rfloor-1\\
                                    \end{array}
                                    \Big),
\end{eqnarray}
Using the relation $\lfloor x \rfloor= x+O(1)$ combined with \eqref{CCatalan} and \eqref{CCCatalan}, we obtain
$$ x \Big(\begin{array}{c}
                                      2(\lfloor x \rfloor-1)\\
                                      \lfloor x\rfloor-1\\
                                    \end{array}
                                    \Big)
\leq s(t_x) \leq x \log(x)
\Big(\begin{array}{c} 2(\lfloor x \rfloor-1)\\
                                      \lfloor x\rfloor-1\\
                                    \end{array}
                                    \Big).
$$
By Lemma \ref{Catalan2}, this gives
$$x^{\frac34}\frac{4^{\sqrt{x}-1}}{\sqrt{\pi}} \leq s(t_x)\leq x^{\frac34}
\log(x)\frac{4^{\sqrt{x}-1}}{\sqrt{\pi}},$$
which proves the theorem.
\end{proof}
Now we are able to give the proof of Theorem \ref{mainth1}.
\begin{proof}[of Theorem \ref{mainth1}]{}By Proposition \ref{equiva1}, it is sufficient to show that
 $$s(t_n)\log^2(s(t_n))\sim(n-2)s(t_n).$$
For that, observe that we have
\begin{eqnarray*}
\frac{s(t_n)\log^2(s(t_n))}{(n-2)s(t_n)}&=&\frac{\log^2(s(t_n))}{n-2}
\end{eqnarray*}
Applying Theorem \ref{mainth2}, we deduce that
\[
\log(s(t_n)) \sim \log(4)~\sqrt{n}.
\]
Whence
\[
\log^2(s(t_n)) \sim \log^2(4)~n.
\]
Hence
$$\frac{\log^2(s(t_n))}{(n-2)} \tendn \log^2(4).$$
We deduce that
$$\displaystyle \frac{\omega(t_n)}{s(t_n)\log^2 s(t_n)}\tendn \frac1{\log^2(4)}<1.$$
This finishes the proof of the theorem.
\end{proof}
%====================================================================================================================================
\section{Conclusion}\label{Conclusion}
%===============================================================================
In this paper we show how binary trees can be used to design a fast algorithm for computing an automaton with a reduced\footnote{Asymptotically minimal.} number of transitions recognizing the language $L(E_n)$. We have verify that our algorithm gives the minimal number of transitions for $n=1 \mbox{ to } 7$ (see Table~\ref{tab:t1}) and we have shown that our reduction is asymptotically a minimization. Hence, we conjecture that Algorithm~\ref{alg:alg4} computes the minimal transition automaton. %On the other hand, we are convinced that tree structures can be used in the study of prime numbers and may also used to establish more precise lower bound complexity.
%\footnote{An interesting way to show this conjecture is to show that $w(t_n) \in O(\frac{s(t_n)\log^2 s(t_n)}{\log\log~s(t_n)})$.}
%====================================================================================================================================
%\begin{note}
%Given a prime number $p$, we deduce that the prime tree $t_p$ which has $|M_{t_p}|=2$ is associated with the $p^{th}$ Eratosthenes-Pascal's Triangle's row which has also $S_p=2$. We show that we can associate a $P$-tree to every natural number, and, we prove that the prime trees (a subclass of $P$-trees associated with prime numbers), are easily identifiable. Therefore, if we have the prime tree $t_p$ associated with the prime number $p$ the computation of the next prime number can be realized using Algorithm~\ref{alg:alg50}.
%\begin{algorithm}
%\caption{Successor$(t_p)$}
%\label{alg:alg50}
%\begin{algorithmic}[1]
%\REQUIRE $p \mbox{ a prime }$
%\ENSURE $t_{Successor(p)}$
%\REPEAT \STATE $t_p \leftarrow SplitAll(t_p)$ \UNTIL{$|M_{t_p}|=2$}
%\end{algorithmic}
%\end{algorithm}
%%From prime trees, several problems can be investigated as future work. The first is to find a better compact representation of $P$-trees, and to study its time and space complexity. The second is to characterize a subclass of $P$-trees or a superclass of prime trees.
%\end{note}
%===============================================================================
\subsubsection*{Acknowledgments.}Special thanks to Sa\"id Abdedda\"im, Alexis B\`es, Patrick C\'egielski, Jean-Marc Champarnaud and Yuri Matiyasevich.
%===============================================================================
\bibliographystyle{plain}
\bibliography{Bibliographie}
%===============================================================================
%\newpage
\appendix
%===================================================================================================================================
\section{Asymptotic notations}\label{AsymptoticNotations}
%===================================================================================================================================
Following \cite{Flajolet}, we employ the standard asymptotic notation called Bachmann–Landau notation as follows.
Let $\Se$ be a set and $s_0 \in \Se$ a particular element of $\Se$. We assume a notion of
neighbourhood to exist on $\Se$. Examples are $\Se = \Z_{>0} \bigcup \{+\infty\}$ with
$s_0 = +\infty$, $\Se =\R$
with $s_0$ any point in $\R$; $\Se = \C$ or a subset of $\C$ with $s_0 = 0$, and so on. Two functions
$f$ and $g$ from $\Se \setminus \{s_0\}$ to $\R$ or $\C$ are given.
\begin{itemize}
\item $O$-notation: write
$$f(s) \stackrel{s\rightarrow s_0}{=}O(g(s)),$$
\noindent{} if the ratio $\displaystyle \frac{f(s)}{g(s)}$ stays bounded as $s\rightarrow s_0$ in $\Se$. In other words, there
exists a neighborhood $V$ of $s_0$ and a constant $C > 0$ such that
$$|f(s)| < C |g(s)|,~~ s \in V, s \neq s_0.$$
One also says that "$f$ is of order at most $g$", or "$f$ is big-Oh of $g$"(as s
tends to $s_0$).
         \item $o$-notation: write
$$f(s) \stackrel{s\rightarrow s_0}{=} o(g(s)),$$
\noindent{} if the ratio $\displaystyle \frac{f(s)}{g(s)}$ tends to $0$ as $s\rightarrow s_0$ in $\Se$. In other words, for any
(arbitrarily small) $\varepsilon > 0$, there exists a neighborhood $\mathcal{V}_{\varepsilon}$ of $s_0$ (depending
on $\varepsilon$), such that
$$|f(s)| <\varepsilon |g(s)|,\quad\quad s \in \mathcal{V}_{\varepsilon}, s \neq s_0.$$
One also says that "$f$ is of order smaller than $g$, or $f$ is little-oh of $g$" (as $s$
tends to $s_0$).
\item $\sim$-notation: write
$$ f(s) \stackrel{s\rightarrow s_0}{\sim}g(s),$$
if the ratio $\displaystyle \frac{f(s)}{g(s)}$ tends to $1$ as $s \rightarrow s_0$ in $\Se$. One also says that "$f$ and
$g$ are asymptotically equivalent" (as $s$ tends to $s_0$).
 \item $\Omega$-notation: write
$$f(s) \stackrel{s\rightarrow s_0}{=} \Omega(g(s)),$$
\noindent{}if the ratio $\displaystyle \frac{f(s)}{g(s)}$ stays bounded from below in modulus by a non-zero
quantity, as $s \rightarrow s_0$ in $\Se$. Which means that there exists $k>0$ and a neighborhood $\mathcal{V}$ of $s_0$, such that
$$f(s) \geq k.g(s),~~ s \in \mathcal{V}.$$

 One then says that $f$ is of order at least $g$.
         \item $\theta$-notation: if $f(s)= O(g(s))$ and $f(s) =\Omega(g(s))$, write
$$ f(s) \stackrel{s\rightarrow s_0}{=} \theta(g(s)).$$
This implies that there exits $k,C>0$ and a neighborhood $\mathcal{V}$ of $s_0$, such that
$$k. g(s) \leq f(s) \leq C.g(s),~~ s \in \mathcal{V}.$$
One then says that $f$ is of order exactly $g$.
\end{itemize}
%For instance, one has as $n \in  \Z_{>0}$:
%$$log(n) = o(\sqrt{n});~~;\big(\begin{array}{c}
%                                n \\
%                                2
%                              \end{array}\big)=\Omega(n\sqrt{n});~~
% n +\pi \sqrt{n}=\theta(n).$$
At this point we are able to make a parallel between the history of our contribution and the history of the famous Prime Number Theorem (PNT) which we shall use later in its weaker form. The PNT Theorem concerns the asymptotic behavior of the prime-counting function $\pi(x)=\left|\{p \leq x, p {\textrm{~~prime}}\}\right|$. Using asymptotic notation the PNT can be restated as
$$\pi(x)\sim\frac{x}{\ln x}.\! $$
The behavior of $\pi(x)$ has been the object of intense study by many celebrated mathematicians ever since the eighteenth centry.
Inspection of tables of primes led Gauss (1792) and Legendre (1798) to conjecture the PNT. In
1808 Legendre published the formula $\pi(x) = x/(log x + A(x))$, where $A(x)$ tends to a constant $B=-1.08366$ as $x \longrightarrow +\infty$, which means that $\pi$ is $\Omega(x/log(x))$.\\
According to Bateman and Diamond \cite{Bateman-Diamond}, The first person to establish the true order of $\pi(x)$ was P. L. Chebyshev. Indeed, in two papers from 1848 and 1850, Chebychev prove that $\pi(x)$ is $\theta(x/log(x)$. This result is known in nowadays as Chebychev Theorem.\\
Finally, in 1896 the PNT was first proved by Hadamard and de la Vall\'ee Poussin. Their proofs were long and intricate. A simplified modern presentation is given on pages 41-47 of Titchmarsh's book on the Riemann Zeta function \cite{Titchmarsh}.
\end{document}